\documentclass[11pt,reqno]{article}   
\usepackage{fullpage}

\usepackage[utf8]{inputenc}
\usepackage[T1]{fontenc}
\usepackage{amsmath}
\usepackage{amssymb}
\usepackage{graphicx}
\usepackage[left=3.00cm, right=3.00cm]{geometry}

\usepackage{amsthm}

\theoremstyle{theorem}

\newtheorem{thm}{Theorem}
\newtheorem{prop}[thm]{Proposition}
\newtheorem{lem}[thm]{Lemma}

\newtheorem*{thm*}{Theorem}

\theoremstyle{definition}
\newtheorem{defn}[thm]{Definition}
\newtheorem{rmk}[thm]{Remark}
\newtheorem{eg}[thm]{Example}

\usepackage{subcaption}

\usepackage{tikz}
\usepackage{tikz-cd}

\usetikzlibrary{decorations.pathreplacing,decorations.markings,angles,quotes}

\usepackage{euscript}
\usepackage{hyperref}
\usepackage{bbold}
\def\on{\operatorname}

\def\sf{\mathsf}

\def\bb{\mathbb}
\def\CC{\bb{C}}
\def\RR{\bb{R}}
\def\NN{\bb{N}}
\def\ZZ{\bb{Z}}

\def\qQ{\mathcal{Q}}
\def\pP{\mathcal{P}}
\def\cC{\mathcal{C}}
\def\mM{\mathcal{M}}

\def\lL{\mathcal{L}}

\def\bB{\mathcal{B}}

\newcommand{\one}{\mathbb{1}}
\newcommand{\ket}[1]{\ensuremath{ | #1 \rangle }}
\newcommand{\bra}[1]{\ensuremath{ \langle #1 | }}
\newcommand{\set}[1]{\ensuremath{ \lbrace #1 \rbrace }}

\usepackage{mathtools}

\usetikzlibrary{calc}

\makeatletter
\newcommand{\gettikzxy}[3]{%
	\tikz@scan@one@point\pgfutil@firstofone#1\relax
	\edef#2{\the\pgf@x}%
	\edef#3{\the\pgf@y}%
}
\makeatother

\newcommand{\PortBlockHalfSize}{2.4}

\newcommand{\Block}[4]{%
	\node (#4) at (#1,#2)
	[
		draw,
		thick,
		minimum width=\PortBlockHalfSize em,
		minimum height=\PortBlockHalfSize em
	]
	{#3};%
}

\newcommand{\VOut}[3]{%
	($(#1.south west)!{#3/(#2+1)}!(#1.south east)$)%
}

\newcommand{\VIn}[3]{%
	($(#1.north west)!{#3/(#2+1)}!(#1.north east)$)%
}

\newcommand{\HIn}[3]{%
	($(#1.north west)!{#3/(#2+1)}!(#1.south west)$)%
}

\newcommand{\HOut}[3]{%
	($(#1.north east)!{#3/(#2+1)}!(#1.south east)$)%
}

\newcommand{\HDraw}[8]{%
	\draw let
		\p1=\HOut{#1}{#2}{#3},
		\p2=\HIn{#4}{#5}{#6}
	in
		(\x1,\y1) to[out=#7,in=#8] (\x2,\y2);
}

\newcommand{\VDraw}[8]{%
	\draw[dashed] let
		\p1=\VOut{#1}{#2}{#3},
		\p2=\VIn{#4}{#5}{#6}
	in
		(\x1,\y1) to[out=#7,in=#8] (\x2,\y2);
}

\newcommand{\VIndraw}[4]{%
	\draw[dashed]
	let
		\p1=\VIn{#1}{#2}{#3}
	in
		(\x1,\y1) to ($(\x1,\y1)+(0,#4)$);
}

\newcommand{\VOutdraw}[4]{%
	\draw[dashed]
	let
		\p1=\VOut{#1}{#2}{#3}
	in
		(\x1,\y1) to ($(\x1,\y1)+(0,-#4)$);
}

\newcommand{\HIndraw}[4]{%
	\draw
	let
		\p1=\HIn{#1}{#2}{#3}
	in
		(\x1,\y1) to ($(\x1,\y1)+(-#4,0)$);
}

\newcommand{\HOutdraw}[4]{%
	\draw
	let
		\p1=\HOut{#1}{#2}{#3}
	in
		(\x1,\y1) to ($(\x1,\y1)+(#4,0)$);
}

\tikzcdset{
  oriented/.style={
    no head,
    postaction={decorate},
    decoration={
      markings,
      mark=at position 0.55 with {
        \arrow{Stealth[length=2mm]}
      }
    }
  }
}

\newcounter{commcount}
\usepackage{xcolor}
\usepackage{todonotes}

 \usepackage{authblk}

\title{
{Double categories for adaptive quantum computation}} 
\begin{document}

\author[1]{Cihan Okay\footnote{cihan.okay@bilkent.edu.tr}}
\author[2]{Walker Stern}
\author[3]{Redi Haderi}
\author[4]{Selman Ipek}

\affil[1]{{\small{Department of Mathematics, Bilkent University, Ankara, Turkey}}}
\affil[2]{\small{Department of Mathematics, Technical University of Munich, Munich, Germany}}
\affil[3]{{\small{Department of Mathematics, Yildirim Beyazit University, Ankara, Turkey}}}
\affil[4]{{\small{Institute of Theoretical Physics,  Leibniz University Hannover, Germany}}}

	\maketitle

\begin{abstract}  
Quantum computation admits several models that emphasize different
computational primitives and forms of classical control. We develop a
unified double categorical framework for describing these models and the
conversions between them. The syntax is provided by \emph{double port
graphs}, whose horizontal wires carry quantum information and whose vertical
wires carry classical information and control. For each set of port labels,
these graphs form a double category, and this construction is functorial in
the label set. The semantics is given by the one-object  double
category of \emph{adaptive instruments}. Its associated horizontal and
vertical monoidal categories recover, respectively, 
quantum channels and stochastic maps. An assignment of an adaptive instrument to each primitive label therefore
extends canonically to a double functor on labeled double port graphs,
providing their computational semantics.
We apply this framework to prominent models of quantum computation, including the circuit model, measurement-based quantum
computation, quantum computation with magic states, and measurement-based
Pauli computation. Gadget constructions from quantum computing that
implement conversions between these models become double functors. Finally,
we show that the interaction between quantum operations and affine classical
control in measurement-based Pauli computation realizes every Boolean
function in the vertical direction, thereby providing the non-affine
classical operations required for its simulation of the circuit model.
\end{abstract}

\tableofcontents	
	
\section{Introduction}

Quantum computation can be formulated through several models, each
emphasizing different computational primitives and forms of classical
control. Understanding the structural relationships among these models is
important for comparing how they represent and implement quantum
computations. Category theory provides a systematic language for such
comparisons, allowing common structures and differences to be expressed at
a high level of abstraction. For example, categorical quantum mechanics
describes quantum processes through string diagrams
\cite{abramsky2009categorical,selinger2007dagger,
heunen2019categories,coecke2018picturing,van2020zx}. This formalism
abstracts essential features of the Hilbert space formulation of quantum
theory and allows fundamental constructions in quantum information
processing, such as quantum teleportation, to be studied compositionally.
The usual string diagrammatic framework is based on a  monoidal
category of processes, in which quantum channels form the primary
morphisms, while classical information and control are encoded internally.
In this work, we separate the classical and quantum directions within a
double categorical framework. We introduce double port graphs as a
two-dimensional syntax for computational processes and adaptive instruments
as their semantics. This separation makes the interaction between quantum
processes and classical control explicit and provides a common setting in
which different models of quantum computation and the conversions between
them can be described by double functors.

\begin{figure}[h]
\centering
\subfloat[]{
\begin{tikzpicture}[x=0.5em,y=0.5em]
\Block{0}{0}{}{Phi}
\VIndraw{Phi}{1}{1}{3}
\VOutdraw{Phi}{1}{1}{3}
\HIndraw{Phi}{1}{1}{3}
\HOutdraw{Phi}{1}{1}{3}
\end{tikzpicture}
\label{fig:a}
\vspace{1.1em}
}
\hspace{8em}
\subfloat[]{
\begin{tikzpicture}[x=0.5em,y=0.5em]

\Block{0}{0}{$\Phi$}{Phi}

\VIndraw{Phi}{1}{1}{3}
\node[above] at ($(Phi.north)+(0,3)$) {${X}$};

\VOutdraw{Phi}{1}{1}{3}
\node[below] at ($(Phi.south)+(0,-3)$) {${Y}$};

\HIndraw{Phi}{1}{1}{3}
\node[left] at ($(Phi.west)+(-3,0)$) {$V$};

\HOutdraw{Phi}{1}{1}{3}
\node[right] at ($(Phi.east)+(3,0)$) {$W$};

\end{tikzpicture}
\label{fig:b}
}
\caption{(a) A double port graph consisting of a single node.  (b) {A double port graph labeled by an} adaptive instrument $\Phi$ with input set ${X}$ and output set ${Y}$ operating on the input Hilbert space $V$ with output Hilbert space $W$. The horizontal direction represents quantum computation, whereas the vertical direction represents classical computation.}
\label{fig:phi-blocks}
\end{figure}

Our philosophy is to adopt a high-level treatment of operations represented by boxes with two kinds of composition. Such a basic component is depicted as in Figure \ref{fig:a}
where solid horizontal wires denote qubits and dashed vertical wires denote bits. Thus, quantum information flows along solid wires in the horizontal direction, while classical information flows along dashed wires in the vertical direction. We formalize this structure as a \emph{double port graph}, which captures the syntax of the two interacting modes of composition independently of their interpretation. A double port graph extends the one-dimensional notion of a port graph \cite{fong2018seven}, which is a graph with distinguished dangling edges for inputs and outputs. A key property of double port graphs is that, when regarded as a one-dimensional object (i.e., interpreting dashed wires as solid), one recovers an ordinary port graph. Treating each box as a vertex yields an acyclic directed graph, which naturally induces a temporal order on the operations, both horizontally and vertically. This temporal order has a causal nature: if an operation is controlled by the output of another, then its own outputs cannot be used as inputs to the controlling box.

We assemble $\lL$-labeled double port graphs into a double category
$\sf{DPG}_{\lL}$, where $\lL$ is a set of port labels. A morphism of port labels
$f\colon\lL\to\mM$ induces a relabeling double functor
\[
\overline{f}\colon\sf{DPG}_{\lL}\longrightarrow\sf{DPG}_{\mM}.
\]
These assignments define a functor
\[
\begin{tikzcd}
\sf{DPG}\colon &[-3em] \sf{PL} \arrow[r] & \sf{DCat}
\end{tikzcd}
\]
from the category of port label sets to the category of double categories
and double functors.
Double port graphs may themselves be used as labels. Writing
$\on{Sq}(\sf{D})$ for the set of squares of a double category $\sf{D}$,
equipped with its horizontal and vertical source and target maps, we obtain
a pasting double functor
\[
\begin{tikzcd}
\on{paste}_{\lL}\colon &[-3em]
\sf{DPG}_{\on{Sq}(\sf{DPG}_{\lL})}
\arrow[r]
&
\sf{DPG}_{\lL}.
\end{tikzcd}
\]
This functor replaces each vertex of a double port graph by the
$\lL$-labeled double port graph assigned to its label and then connects the
resulting ports according to the wiring of the ambient graph.
Consequently, a morphism of port labels
$
f\colon\mM\to\on{Sq}(\sf{DPG}_{\lL}),
$
which assigns an $\lL$-labeled double port graph to every $\mM$-label,
induces the composite double functor
\[
\begin{tikzcd}
\kappa_{\mM,\lL}\colon &[-3em]
\sf{DPG}_{\mM}
\arrow[r,"\overline{f}"]
&
\sf{DPG}_{\on{Sq}(\sf{DPG}_{\lL})}
\arrow[r,"{\on{paste}_{\lL}}"]
&
\sf{DPG}_{\lL}.
\end{tikzcd}
\]
Thus, an assignment of labels to compatible diagrams extends
functorially to a substitution of entire labeled double port graphs. This
construction provides the categorical basis for the
conversions between computational models developed later in the paper.

Quantum operations, including unitary transformations and measurements, can
be described by \emph{instruments}
\cite{davies1970operational,watrous2018theory}. An instrument is a
collection of completely positive linear maps, indexed by an output set,
whose sum is a channel. We introduce an input-indexed version of this
notion, called an \emph{adaptive instrument}. Given finite sets $X$ and $Y$ and finite-dimensional
Hilbert spaces $V$ and $W$, an adaptive instrument $\Phi$ consists of a
family
\[ 
\begin{tikzcd}
\Phi_a^b\colon &[-3em] L(V) \arrow[r] & L(W)
\end{tikzcd},
\qquad a\in X,\quad b\in Y,
\]
of completely positive maps such that, for each $a\in X$, the family
$\{\Phi_a^b\}_{b\in Y}$ is an instrument; see Figure~\ref{fig:b}. Thus, the classical input
$a$ selects the instrument to be applied, while $b$ corresponds to its output.
A typical example is an adaptive quantum measurement whose measurement
angle is determined by the output of a preceding measurement.

Given appropriately composable adaptive instruments $\Phi$ and $\Psi$, we
define their horizontal and vertical compositions, denoted respectively by
\[
\Psi\circ\Phi
\qquad\text{and}\qquad
\Psi\bullet\Phi.
\]
Horizontal composition is induced by the composition of completely positive
maps, whereas vertical composition is defined using their tensor product
together with summation over the intermediate classical outputs. We prove that these operations satisfy the interchange law, the key property of a double category governing the interaction between horizontal and vertical composition. Adaptive
instruments therefore form a one-object double category
$\sf{Inst}$.
The associated horizontal monoidal category is identified with the strict
monoidal category $\sf{Chan}_{\NN}$ of quantum channels between the Hilbert
spaces $\CC[\underline{n}]$, while the associated vertical monoidal category
is identified with the strict monoidal category $\sf{Stoch}_{\NN}$ of
stochastic maps between the finite sets $\underline{n}$
\cite{watrous2018theory,selinger2007dagger,jacobs2010convexity}. In this way,
$\sf{Inst}$ provides the semantics of our framework, with quantum processes
and classical control represented by its two categorical directions.

Our treatment of computational models begins with the double category
$\sf{QBit}$ of qubit instruments. This is the full double subcategory of
$\sf{Inst}$ whose Hilbert-space boundaries are tensor powers
$(\CC^2)^{\otimes n}$ and whose classical boundaries are the bit-string
sets $\ZZ_2^n$. 
Specifying a qubit instrument for each label in $\lL$ determines a double functor. 
\[
\begin{tikzcd}
\phi_\lL\colon &[-3em]
\sf{DPG}_{\lL} \arrow[r] &
\sf{QBit}.
\end{tikzcd}
\]
For each computational model considered below, the corresponding choice of
label set and instrument assignment yields a double functor of this
form. 
Our formulation of quantum computation begins with the circuit model (QC)
\cite{nielsen2010quantum} and then extends to measurement-based quantum
computation (MBQC) \cite{raussendorf2001one} and quantum computation with
magic states (QCM) \cite{bravyi2005universal}. We write $\cC$ for the QC
labels, $\mM$ for the MBQC labels, and $\qQ$ for the QCM labels. These
fundamental models are represented as double categories of labeled port
graphs:
\[
\sf{DPG}_\cC, \qquad \sf{DPG}_\mM,\qquad \sf{DPG}_\qQ.
\]
The label sets specify the computational primitives of each model, from
which all other allowed operations can be constructed by horizontal and
vertical composition of labeled boxes.
A computation in one model can be translated into an equivalent computation
in another, where equivalence means that the associated instruments compose
to the same overall operation. Formally, these conversions, together with
the corresponding instrument assignments, are captured by double functors,
giving the commutative diagram
\begin{equation}\label{dia:Q-C-M}
\begin{tikzcd}[column sep=huge, row sep=large]
  & \sf{DPG}_\qQ \arrow[dr,"\phi_{\qQ}"] &\\
\sf{DPG}_\cC  \arrow[rr,"\phi_\cC"] \arrow[ru,"\kappa_{\cC,\qQ}"]\arrow[dr,"\kappa_{\cC,\mM}"'] && \sf{QBit} \\
&\sf{DPG}_{\mM} \arrow[ru,"\phi_\mM"']&
\end{tikzcd}
\end{equation}
A model conversion is defined using the pasting operation on labeled double
port graphs and is typically specified by a collection of special
\emph{gadgets} that implement the instrument associated with a label in one
model using labels from another model.

In practice, quantum computation requires the approximate implementation of
arbitrary unitary maps, a property known as \emph{quantum universality}
\cite{kitaev2002classical}. The principal computational primitive in the
circuit label set $\cC$ is the $T$-gate. By contrast, neither MBQC nor QCM
contains the $T$-gate as a primitive label. 
In MBQC, a typical adaptive operation is a measurement performed at an angle
$\pm\alpha$ relative to the Pauli $X$ axis in the equatorial plane of the
Bloch sphere. The computational resource corresponding to the $T$-gate is
provided by measurements with $\alpha=\pi/4$. In QCM, by contrast, the computational resource is supplied through the preparation of the magic state
$T\ket{+}$, obtained by rotating the $X$-eigenstate $\ket{+}$ by the angle
$\pi/4$. In both adaptive models, classically controlled unitary corrections
may be applied to ensure the deterministic implementation of the desired
unitary map. In this way, adaptivity compensates for the probabilistic nature
of quantum measurements.

At the intersection of MBQC and QCM lies an interesting model known as
measurement-based Pauli computation (MBPC)
\cite{okay2024classical,danos2007pauli}. In MBPC, whose label set is denoted
by $\pP$, the state preparations are inherited from QCM, while the
entanglement structure follows that of MBQC, with measurements restricted to
the Pauli $X$ and $Y$ observables (i.e., $\alpha =0$ or $ \pi/2$). To ensure the
deterministic implementation of unitary maps, we extend the label set to
$\widetilde\pP$ by including correction boxes. The corresponding model
conversions and instrument assignments yield the following commutative
diagram:
\begin{equation}\label{dia:Q-C-P}
\begin{tikzcd}[column sep=huge, row sep=large]
&\sf{DPG}_\cC \arrow[d,"\kappa_{\cC,\widetilde{\pP}}"'] \arrow[rr,"\phi_\cC"] && \sf{QBit} \\
\sf{DPG}_{\pP}  \arrow[r,hook]&\sf{DPG}_{\widetilde \pP} \arrow[rru,"\phi_{\widetilde{\pP}}"]\arrow[rr,"\kappa_{\widetilde{\pP},\qQ}"'] && \sf{DPG}_{\qQ} \arrow[u,"\phi_\qQ"'] \\ 
\end{tikzcd}
\end{equation}
This completes the family of computational models considered in this paper.
More generally, additional model conversions can be formulated between
different pairs of these computational models.

The principal construction behind the conversion from the circuit model to
MBPC is the implementation of the $T$-gate using a gadget known as the
teleported $J(\pi/4)$-gadget
(Figure~\ref{fig:teleported-j-pi-over-4}). This gadget requires the
non-affine Boolean operation $\on{AND}$, although the classical labels in
$\pP$ generate only affine Boolean functions. This apparent obstruction
reveals an additional feature of our two-dimensional compositional
framework. Using a Greenberger--Horne--Zeilinger (GHZ) state, adaptive Pauli
measurements, and affine Boolean control, the OR-gadget realizes the
non-affine Boolean operation $\on{OR}$ entirely within MBPC
\cite{anders2009computational}; see
Definition~\ref{def:Anders Browne gadget}. Since the affine Boolean operations together with $\on{OR}$ generate all Boolean functions, which constitute the morphisms of the full subcategory $\sf{Bool}$ of the category of sets, every deterministic Boolean map can be realized in the vertical direction of $\sf{MBPC}$, the image of the double functor $\phi_\pP$. The same
construction supplies the non-affine Boolean operations required by the
circuit-to-MBPC conversion. Thus, the interaction with the horizontal
quantum direction enlarges the vertical classical structure from affine
Boolean computation to arbitrary Boolean computation, a transition closely
connected to quantum contextuality in measurement-based models
\cite{raussendorf2013contextuality,abramsky2017contextual,
raussendorf2016cohomological,abramsky2024combining}.

\begin{thm*} 
Let $\sf{MBPC}$ denote the image of the double functor $\phi_\pP$. Then,
\[
\sf{Bool}
\subset
\sf{V}(\sf{MBPC}).
\]
\end{thm*}

The paper is organized as follows. Section~\ref{sec:double cat of port graphs}
introduces double port graphs, our main syntactic framework and a
two-dimensional generalization of ordinary port graphs.
Section~\ref{sec:double cat of inst} develops the double category of adaptive
instruments, which provides the semantics of the framework. It also presents
double port graph descriptions of Boolean and quantum circuit models,
constructs their associated semantics double functors, and introduces
adaptive local instruments in preparation for the study of adaptive quantum
computation. Section~\ref{sec:adaptive quantum computation} treats the main
adaptive computational models considered in the paper, including
measurement-based quantum computation, quantum computation with magic
states, and measurement-based Pauli computation, together with the
conversions between them and their computational properties.
Appendix~\ref{sec:double category} reviews the necessary background on double
categories, while Appendix~\ref{sec:standard form} derives the standard form
of MBQC within our double categorical framework.

\paragraph{Acknowledgments.}
This work is supported by the Air Force Office of Scientific Research (AFOSR) under award number  FA9550-24-1-0257. The first author also acknowledges support from the Digital Horizon Europe project FoQaCiA, GA no. 101070558.

\section{{Double port graphs}}
\label{sec:double cat of port graphs}	

{In this section, we introduce \emph{double port graphs}. Intuitively, these are special kinds of graphs equipped with both horizontal and vertical inputs and outputs, allowing composition in two distinct directions—horizontal and vertical—thus forming a double category. We also define a labeled version, in which each node is assigned a label from a specified set. This section provides the technical groundwork used to formalize various quantum computational models developed in the following sections.   {Throughout the paper, a double category means a strict double category in the sense of Definition~\ref{def:double category} in Appendix~\ref{sec:double category}.}

	\subsection{Port graphs}
 	\label{sec:port graphs}

In what follows, we will use the term \emph{port graph}, though our notion differs from that of Spivak \cite{fong2018seven} in that the input and output labels of a vertex form an \emph{ordered} set. {Let us begin by fixing some notation.} We will denote by $\sf{Ord}$ and $\sf{Fin}$  
the category of linearly ordered finite sets and the category of finite sets. 
{W}e will denote by $\NN$ the set of natural numbers, including $0$, that is, the set of isomorphism classes in $\sf{Fin}$ or $\sf{Ord}$. 
{The skeletal versions of $\sf{Ord}$ and $\sf{Fin}$ will be denoted by $\sf{Ord}_\NN$ and $\sf{Fin}_\NN$, respectively, and their objects will be denoted} 
by $\underline{n}=\{1,\ldots,n\}$ for $n\geq 0$, {with the convention that $\underline 0$ is the empty set}.

	\begin{defn}
		For $n,m\in \NN$, define an \emph{$(m,n)$ pre-port graph} $\Gamma$ to consist of the following data.
		\begin{itemize}
			\item A {finite} set $X$ called the \emph{vertices} of $\Gamma$. 
			\item A pair of functions
			\[
			\begin{tikzcd}
				\on{in}:&[-3em] X \arrow[r] &\sf{Ord} 
			\end{tikzcd} \quad \begin{tikzcd}
				\on{out}:&[-3em] X \arrow[r] &\sf{Ord} 
			\end{tikzcd}
			\]
			which assign to each vertex sets of inputs and outputs, respectively. We will write 
			\[
				O:= \coprod_{x\in X} \on{out}(x),\qquad
				I:= \coprod_{x\in X} \on{in}(x) 
			\]
			(where the coproduct is taken in $\sf{Fin}$, i.e., is the disjoint union of sets) for the sets of all inputs and outputs, respectively.
			\item A bijection
			\[
			\begin{tikzcd}
				\iota: &[-3em] \underline{m}\amalg O\arrow[r] & I\amalg \underline{n}.
			\end{tikzcd}
			\]
		\end{itemize}
		An \emph{isomorphism of  $(m,n)$-port graphs} consists of a bijection of the vertex sets and natural bijections between the in- and out-maps, such that the induced map on sets of all inputs and sets of all outputs commutes with the $\iota$'s.
	\end{defn}
	
	To make sense of this definition, we must define the internal flow graph.

	\begin{defn}
		Let $(X,\on{out},\on{in},\iota)$ be an $(m,n)$ pre-port graph. The corresponding \emph{internal flow graph} is the oriented graph with vertices $X\amalg\underline{n}\amalg \underline{m}$, set of edges $\underline{m}\amalg O$, and source and target maps given by 
		\[
		s(e) =\begin{cases}
			x & e= \on{out}(x)\\
			e & e\in \underline{m}
		\end{cases}
		\]
		and 
		\[
		t(e) =\begin{cases}
			x & \iota(e)=\on{in}(x) \\
			\iota(e) & \iota(e)\in \underline{n}. 
		\end{cases}
		\]
	\end{defn}

	\begin{defn}\label{def:port graph}
		An \emph{$(m,n)$-port graph} is an $(m,n)$-pre-port graph satisfying the condition that the internal flow graph is acyclic. 
	\end{defn}

\begin{eg}\label{ex:port-graph}
{Here is a {$(2,3)$}-port graph} {where the vertices are depicted as boxes:}	
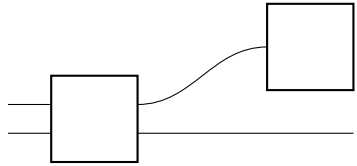
\begin{center}
\begin{tikzpicture}[x=0.5em,y=0.5em]

  \Block{10}{0}{$a$}{Main}

  \Block{25}{5}{$b$}{Small}

  \HIndraw{Main}{2}{1}{5}   
  \HIndraw{Main}{2}{2}{5}   
 
  \HOutdraw{Main}{3}{3}{20}  
  \HOutdraw{Small}{2}{1}{5}
  \HOutdraw{Small}{2}{2}{5}

  \HDraw{Main}{3}{1}{Small}{3}{1}{0}{180}
  \HDraw{Main}{3}{2}{Small}{3}{2}{0}{180}

  \node[font=\footnotesize, left] at ($(Main.center)+(-8,1)$) {$1$};
  \node[font=\footnotesize, left] at ($(Main.center)+(-8,-1)$) {$2$};

  \node[font=\footnotesize, above] at ($(Main.center)+(-4.2,1)$) {$1$};
  \node[font=\footnotesize, below] at ($(Main.center)+(-4.2,-1)$) {$2$};

  \node[font=\footnotesize, above] at ($(Main.center)+(4.2,1.5)$) {$1$};
  \node[font=\footnotesize, above] at ($(Main.center)+(4.8,-1.7)$) {$2$};
  \node[font=\footnotesize, below] at ($(Main.center)+(4.2,-1.5)$) {$3$};

  \node[font=\footnotesize, right] at ($(Main.center)+(23,-1.5)$) {$3$};

  \node[font=\footnotesize, above] at ($(Small.center)+(-4.2,1.5)$) {$1$};
  \node[font=\footnotesize, below] at ($(Small.center)+(-4.2,0)$) {$2$};

  \node[font=\footnotesize, above] at ($(Small.center)+(4.2,1)$) {$1$};
  \node[font=\footnotesize, below] at ($(Small.center)+(4.2,-1)$) {$2$};

  \node[font=\footnotesize, right] at ($(Small.center)+(8,1)$) {$1$};
  \node[font=\footnotesize, right] at ($(Small.center)+(8,-1)$) {$2$};

\end{tikzpicture}
\end{center}
{The isomorphism $\iota$ specifies how the wires are connected. For example, $\iota(a,1)=(1,b)$, $\iota(a,2)=(b,2)$, and $\iota(a,3)=3$. The internal flow graph is given by} 
\begin{center}
\begin{tikzpicture}[
  x=1em,y=1em,
  >=stealth,
  vtx/.style={circle, fill, inner sep=1.6pt},
  lab/.style={font=\footnotesize}
]

  \node[vtx, label=left:{\footnotesize $1$}] (in1) at (0,1.5) {};
  \node[vtx, label=left:{\footnotesize $2$}] (in2) at (0,-1.5) {};

  \node[vtx, label=above:{$a$}] (a) at (5,0) {};
  \node[vtx, label=above:{$b$}] (b) at (10,1.5) {};

  \node[vtx, label=right:{\footnotesize $1$}] (out1) at (15,2.5) {};
  \node[vtx, label=right:{\footnotesize $2$}] (out2) at (15,0.5) {};
  \node[vtx, label=right:{\footnotesize $3$}] (out3) at (15,-2) {};

  \draw[->] (in1) -- node[above, lab] {$1$} (a);
  \draw[->] (in2) -- node[below, lab] {$2$} (a);

  \draw[->] (a) to[bend left=12]
    node[above, lab] {${a,1}$}
    (b);

  \draw[->] (a) to[bend right=12]
    node[below, lab] {${a,2}$}
    (b);

  \draw[->] (a) to[bend right=20]
    node[below, lab] {${a,3}$}
    (out3);

  \draw[->] (b) -- node[above, lab] {${b,1}$} (out1);
  \draw[->] (b) -- node[below, lab] {${b,2}$} (out2);

\end{tikzpicture}
\end{center}
\end{eg}

	We define a concatenation operation on port graphs as follows. Let $\Gamma=(X, \on{in}^X,\on{out}^X,\iota^X)$ be an $(m,n)$ port graph, and let $\Xi=(Y, \on{in}^Y,\on{out}^Y,\iota^Y)$ be an $(n,k)$ port graph. We define a $(m,k)$ pre-port graph $\Xi\circ \Gamma$ as follows. 
	\begin{itemize}
		\item The vertex set of $\Xi\circ \Gamma$ is $Y\amalg X$. 
		\item The in- and out- functions are determined by the universal property of pushout. 
		\item The bijection $\iota$ is determined by the commutative diagram 
		\[
		\begin{tikzcd}
			\underline{m}\amalg O \arrow[r,phantom,"\cong"{description}] &[-2em] \underline{m}\amalg O^X \amalg O^Y \arrow[r]\arrow[d,hookrightarrow] &[2em] I^X\amalg I^Y \amalg \underline{k}\arrow[r,phantom,"\cong"{description}] &[-2em] I\amalg \underline{k}\\
			& \underline{m}\amalg O^X\amalg \underline{n}\amalg O^Y\arrow[d,"\iota^X\amalg \on{id}"'] & & \\
			& I^X\amalg \underline{n} \amalg \underline{n}\amalg O^Y \arrow[r,"{\on{id}\amalg \nabla\amalg \on{id}}"'] & I^X \amalg \underline{n}\amalg O^Y \arrow[uu,"\on{id}\amalg \iota^Y"']
		\end{tikzcd}
		\]
		where $\nabla$ is the codiagonal map on $\underline{n}$, that is, the map which acts as the identity on each copy of $\underline{n}$. 
	\end{itemize}
	
\begin{eg}\label{ex:concatenation}
The port graph in Example~\ref{ex:port-graph} is obtained by concatenating the following smaller port graphs by identifying the {corresponding top and bottom wires.}
\begin{center}
\begin{tikzpicture}[x=0.5em,y=0.5em]

  \Block{10}{0}{}{Main}

  \HIndraw{Main}{2}{1}{3}
  \HIndraw{Main}{2}{2}{3}

  \HOutdraw{Main}{3}{2}{3}
  \HOutdraw{Main}{3}{1}{3}
  \HOutdraw{Main}{3}{3}{3}
 

  \Block{30}{2}{}{Small}

  \HIndraw{Small}{2}{1}{3}
  \HIndraw{Small}{2}{2}{3}
  \HOutdraw{Small}{2}{1}{3}
  \HOutdraw{Small}{2}{2}{3}

  \draw (27-3,-3) -- (36,-3);

\end{tikzpicture}
\end{center}

\end{eg}	
	\begin{prop}
		The concatenation of an $(m,n)$-port graph with an $(n,k)$ port graph is an $(m,k)$ port graph. The operation of concatenation is associative and unital, with units given by the unique $(n,n)$ port graphs with empty vertex set. Given isomorphisms $\Gamma\cong \Gamma^\prime$ and $\Xi\cong \Xi^\prime$ of port graphs, there is a unique isomorphism $\Xi\circ \Gamma \cong \Xi^\prime \circ \Gamma^\prime$ which restricts to the original isomorphism. 
	\end{prop}

\begin{defn}
{The \emph{category of port graphs}, denoted by $\sf{PG}$, has objects given by $\NN$ and morphisms given by isomorphism classes of port graphs.}
\end{defn}

	\subsection{Double port graphs}
\label{sec:double port graphs}	
	
	\begin{defn}
		Given $U$ and $V$ in $\sf{Ord}$, the \emph{ordinal sum}  $U\oplus V$ of $U$ and $V$ is the linearly ordered finite set whose underlying set is $U\amalg V$, and whose order is given by 
		\[
		a<b \Leftrightarrow \begin{cases}
			a<_U b & a,b\in U \\
			a<_V b & a,b\in V \\
			a\in U, b\in V & \text{else.}
		\end{cases}
		\] 
		This defines a monoidal structure on $\sf{Ord}$, and so induces a monoidal structure on 
{$\sf{Ord}_\NN$,}		
		we will abuse notation by also denoting the induced monoidal structure on {$\sf{Ord}_\NN$} by $\oplus$, though in the later context we have $\underline{n}\oplus\underline{m}=\underline{n+m}$, rather than $\underline{n}\oplus\underline{m}\cong \underline{n+m}$. 
	\end{defn}

	\begin{defn}
		A \emph{2-fold port graph} consists of a pair of port graphs $(\Gamma_v,\Gamma_h)$ with the same vertex set $X$. We will call a 2-fold port graph $(\Gamma_v,\Gamma_h)$ an {$\left(\begin{smallmatrix}
			& k & \\
			m & & n \\
			& \ell & 
		\end{smallmatrix}\right)$} 2-{fold} port graph {if} $\Gamma_h$ is an $(m,n)$-port graph and $\Gamma_v$ is a $(k,\ell)$-port graph.   
	\end{defn}
	
	Given an  
{$
\left(
\begin{smallmatrix}
& k & \\
m & & n \\
& \ell &
\end{smallmatrix}
\right)
$}	2-fold port graph $(\Gamma_v,\Gamma_h)$ with vertex set $X$, we can define an {$(k+m, \ell+n)$}-pre-port graph $\on{Tot}(\Gamma_v,\Gamma_h)$ by setting 
	\[
	\on{in}(x)=\on{in}_v(x)\oplus \on{in}_h(x)
	\]
	and 
	\[
	\on{out}(x)=\on{out}_v(x)\oplus \on{out}_h(x).
	\] 
We furthermore set  
	\[
	\iota(e)= \begin{cases}
		\iota_v(e) & e\in {\underline{k}\amalg O_v}  \\
		\iota_h(e) & e\in {\underline{m} \amalg O_{h}}.
	\end{cases}
	\] 
	
	\begin{defn}
		A \emph{double port graph} is a 2-fold port graph $(\Gamma_v,\Gamma_h)$  such that $\on{Tot}(\Gamma_v,\Gamma_h)$ is a port graph. 
	\end{defn}
	
	We can concatenate double port graphs $(\Gamma_v,\Gamma_h)$ and $(\Xi_v,\Xi_h)$
\begin{itemize}
\item horizontally  by setting 
	\[
	(\Gamma_v,\Gamma_h){\circ}(\Xi_v,\Xi_h)= (\Gamma_v\amalg \Xi_{v},\Gamma_h\circ \Xi_h), 
	\]

\item  vertically by setting 
\[
	(\Gamma_v,\Gamma_h){\bullet}(\Xi_v,\Xi_h)= (\Gamma_v\circ \Xi_{v},\Gamma_h\amalg \Xi_h) 
	\]
\end{itemize}	
		whenever the right-hand side is well defined.



\begin{defn}\label{def:DPG}
The \emph{double category $\mathsf{DPG}$ of port graphs} has one object, 
horizontal and vertical morphism sets $\mathbb{N}$, 
and squares consisting of isomorphism classes of double port graphs.
\end{defn}

\begin{rmk}\label{rem:labeled port as double}
{We regard labeled port graphs as labeled double port graphs for which either the horizontal or the vertical source and target maps send every label to zero. For such a set $\lL$ of port labels, we write $\sf{PG}_{\lL}$ for the category of $\lL$-labeled port graphs.}
\end{rmk}

\begin{eg}\label{ex:double-port}
{A typical double port graph looks as follows}
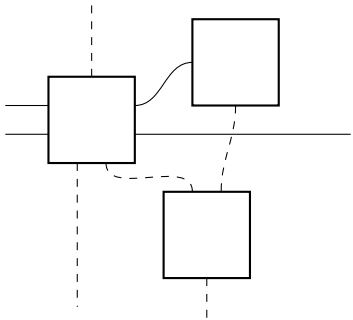
\begin{center}
\begin{tikzpicture}[x=0.5em,y=0.5em]
\def\xshift{12} 

  \Block{\xshift}{0}{}{Main}

  \HIndraw{Main}{2}{1}{3}
  \HIndraw{Main}{2}{2}{3}
  \HOutdraw{Main}{2}{2}{18}

  \VIndraw{Main}{1}{1}{8}
  \VOutdraw{Main}{2}{1}{10}

  \Block{\xshift+10}{4}{}{Up}
  \HDraw{Main}{2}{1}{Up}{1}{1}{0}{180}

  \Block{\xshift+15}{-8}{}{Down}
  \VDraw{Main}{2}{2}{Down}{2}{1}{-90}{90}
  \VDraw{Up}{1}{1}{Down}{2}{2}{-90}{90}

  \VOutdraw{Down}{1}{1}{3}
  \HOutdraw{Down}{1}{1}{3}

\end{tikzpicture}
\end{center}
{We denote horizontal port graphs using solid wires and vertical ones using dashed wires. Vertical concatenation is performed with respect to the dashed wires, whereas horizontal concatenation is performed with respect to the solid wires. As in Example \ref{ex:concatenation}, we first horizontally concatenate the top two boxes and then vertically concatenate the resulting diagram with the bottom box.}
\end{eg}

	\subsection{Port labels}
\label{sec:port labels}	
	
	We now want to generalize our double category $\sf{DPG}$ to labeled double port graphs.

	\begin{defn}
		A \emph{set of port labels} for a double port graph consists of a set ${\lL}$ together with four functions 
		\[
		\begin{tikzcd}
			s^v,t^v,s^h,t^h: &[-3em] {\lL} \arrow[r] & \NN.  
		\end{tikzcd}
		\]
		An \emph{${\lL}$-labeled double port graph} consists of a double port graph $(X,\on{in}_v,\on{out}_v,\on{in}_h,\on{out}_h,\iota_v,\iota_h)$ together with a function 
		\[
		\begin{tikzcd}
			\on{lab}:&[-3em] X\arrow[r] & {\lL}
		\end{tikzcd}
		\]
		such that the diagram 
		\[
		\begin{tikzcd}[column sep=huge, row sep =large]
			\sf{Ord}\arrow[d,"{|-|}"'] & X\arrow[d,"\on{lab}"]\arrow[l,"\on{in}_a"']\arrow[r,"\on{out}_a"] & \sf{Ord}\arrow[d,"{|-|}"] \\
			\NN & {\lL}\arrow[l,"s^a"]\arrow[r,"t^a"'] & \NN
		\end{tikzcd}
		\]
		commutes for each $a\in \{h,v\}$, where $|-|$ denotes the cardinality.  
	\end{defn}

	\begin{eg}\label{ex:square-labels}
{Let $\sf{D}$ be a $1$-object (strict) double category whose horizontal and vertical morphisms are given by $\NN$. Denote by $\on{Sq}(\sf{D})$ the set of squares of $\sf{D}$. The horizontal and vertical source and target maps of $\sf{D}$ endow $\on{Sq}(\sf{D})$ with the structure of a set of port labels. In particular, $\on{Sq}(\sf{DPG})$ can itself be used to label double port graphs. We will employ such labeling in Section~\ref{sec:pasting} when discussing the pasting operation.} 
	\end{eg}

\begin{defn}
Two labeled port graphs are said to be \emph{isomorphic} if there exists an isomorphism of double port graphs between them that commutes with the labeling functions.
\end{defn}

	Notice that, if $(\Gamma_h,\Gamma_v)$ and $(\Xi_h,\Xi_v)$ are double port graphs with vertex sets $X$ and $Y$ respectively, and labeling functions $\on{lab}^X$ and $\on{lab}^Y$ valued in a set of port labels ${\lL}$, the universal property of the coproduct determines a unique labeling 
	\[
	\begin{tikzcd}
		\on{lab}: &[-3em] X\amalg Y \arrow[r] & {\lL}
	\end{tikzcd}
	\]
	on the disjoint union of $\Gamma$ and $\Xi$ compatible with the inclusions. If $\Gamma$ and $\Xi$ are composable (either horizontally or vertically), this yields a canonical ${\lL}$-labeling of the composite. 
	
\begin{defn}
The \emph{double category $\mathsf{DPG}_{{\lL}}$ of ${\lL}$-labeled port graphs} has one object, 
horizontal and vertical morphism sets $\mathbb{N}$, 
and squares consisting of isomorphism classes of ${\lL}$-labeled double port graphs.
\end{defn}

	\subsection{Pasting double port graphs}
\label{sec:pasting}	
	
	We now define a pasting rule that allows us to glue a double port graph with the correct numbers of horizontal and vertical inputs and outputs into the place of a vertex of another double port graph.  
	
	On a formal level, let $\Gamma=(\Gamma_h,\Gamma_v)$ be an 
{$
\left(
\begin{smallmatrix}
& k & \\
m & & n \\
& \ell &
\end{smallmatrix}
\right)
$}-double port graph with vertex set $X$ (and structure morphisms written with a superscript $X$), and let 
	\[
	\begin{tikzcd}
		\on{lab}: &[-3em] X \arrow[r] & \on{Sq}(\sf{DPG}) 
	\end{tikzcd}
	\]
	be a labeling of $\Gamma$ by squares of $\sf{DPG}$. 
	
	We will denote a chosen representative of $\on{lab}(x)$ by $\Lambda^x=(\Lambda^x_h,\Lambda^x_v)$ with set of vertices $Y^x$. The structure maps of $\Lambda^x$ will be decorated with superscripts $x$, e.g., $\iota_h^x$. Similarly, $\Lambda^x$ will be a 
{$
\left(
\begin{smallmatrix}
& k^x & \\
m^x & & n^x \\
& \ell^x &
\end{smallmatrix}
\right)
$}-double port graph. 
	
	Given these data, we define a new 
{$
\left(
\begin{smallmatrix}
& k & \\
m & & n \\
& \ell &
\end{smallmatrix}
\right)
$}-double port graph $\Gamma \bullet \{\Lambda^x\}_{x\in X}$ as follows:
	\begin{itemize}
		\item The set of vertices of $\Gamma \bullet \{\Lambda^x\}_{x\in X}$ is 
		\[
		\coprod_{x\in X} Y^x. 
		\]
		\item For every $x$, and every $y\in Y^x$ we define 
		\[
		\begin{aligned}
			\on{in}_v(y)&=\on{in}^x_v(y)\\ 	
			\on{in}_h(y)&=\on{in}^x_h(y)\\
			\on{out}_v(y)&=\on{out}^x_v(y)\\
			\on{out}_h(y)&=\on{out}^x_h(y).
		\end{aligned}
		\]
		\item The isomorphism $\iota_h$ is defined as follows. Note first that, for every $x\in X$, there is a unique order-preserving bijection $\alpha_x: \on{in}_h(x)\cong \underline{m}^x$, and similarly $\beta_x:\on{out}_h(x)\cong \underline{n}^x$. Define a relation on the set 
		\[
		\coprod_{x\in X} \left(\underline{m}^x\amalg O_h^x\right)
		\]
		as follows. We say that $j\in \underline{m}^x$ is equivalent to $k\in O_h^z$ if $(\iota_h^X)^{-1}\alpha_x^{-1}(j)=\beta_z^{-1}(k)$. Note that, for $j\underline{m}^x$, there is either a unique $z\in X$ and a unique $k\in in O_h^z$ such that $j\sim k$, or  $(\iota_h^X)^{-1}(\alpha(j))=\ell\in \underline{m}$. By Lemma \ref{lem:quotient_set} below, the set $\left(\underline{m}\setminus (\iota^X_h)^{-1}(\underline{n})\right)\amalg \left(\coprod_{x\in X} O_h^x\right)$ is the quotient. 
		
		An identical construction displays $\left(\underline{n}\setminus\iota_h^X(\underline{m})\right)\amalg\left(\coprod_{x\in X}I_h^x\right)$ as a quotient of 
		\[
		\coprod_{x\in X} \left(I_h^x \amalg \underline{n}^x\right). 
		\]
		We then define $\iota_h$ to be the unique map such that the following diagram commutes. 
		\begin{equation}\label{eq:iota_h_diagram}
			\begin{tikzcd}
				& \coprod_{x\in X} \left(\underline{m}^x\amalg O_h^x\right)\arrow[r,"{\coprod_X\iota_h^x}"]\arrow[d] &[2em] \coprod_{x\in X} \left(I_h^x \amalg \underline{n}^x\right)\arrow[d] & \\
				\underline{m}\amalg O_h \arrow[r,phantom,"\cong"{description}]& \underline{m}\amalg\left(\coprod_{x\in X}O_h^x\right)	\arrow[r,"\iota_h"] & \underline{n} \amalg \left(\coprod_{x\in X} I_h^x\right)\arrow[r,phantom,"\cong"{description}] & \underline{n}\amalg I_h\\
				&\underline{m}\cap (\iota_h^X)^{-1}(\underline{n})\arrow[u]\arrow[r,"\iota_h^X"'] &\underline{n}\cap \iota_h^X(\underline{m})\arrow[u] &
			\end{tikzcd}
		\end{equation}
	\end{itemize}
	The vertical structure is defined identically to the horizontal structure.

	\begin{lem}\label{lem:quotient_set}
		The quotient of 
		\[
		\coprod_{x\in X} \left(\underline{m}^x\amalg O_h^x\right)
		\]
		by the relation defined above is canonically isomorphic to $\left(\underline{m}\setminus (\iota^X_h)^{-1}(\underline{n})\right)\amalg \left(\coprod_{x\in X} O_h^x\right)$. 
	\end{lem}
	
	\begin{proof}
		We define a map
		\[
		\begin{tikzcd}
			\pi:&[-3em]	\displaystyle\coprod_{x\in X} \left(\underline{m}^x\amalg O_h^x\right)\arrow[r] & \left(\underline{m}\setminus (\iota^X_h)^{-1}(\underline{n})\right)\amalg \left(\coprod_{x\in X} O_h^x\right)
		\end{tikzcd}
		\]
		as follows. On each copy of $O_h^x$, $\pi$ acts as the identity. For $j\in\underline{m}^x$
		\[
		\pi(j):= \begin{cases}
			k\in O_h^z & k\sim j \\
			(\iota_h^X)^{-1}(\alpha^{-1}_x(j)) & \text{else.}
		\end{cases}
		\]
		This map is a quotient onto its image by construction. Moreover, for every $\ell\in \underline{m}\setminus (\iota^X_h)^{-1}(\underline{n})$, there is some $x\in X$ such that $\iota_h^X(\ell)\in \on{in}_h^X(x)$, and so the map is surjective as well, completing the proof. 
	\end{proof}

	\begin{prop}
		The construction above defines a double port graph $\Gamma\bullet\{\Lambda^x\}_{x\in X}$. 
	\end{prop}
	
	\begin{proof}
		We will show that the horizontal structure forms a port graph, and leave the remaining checks to the reader. Only two facts need to be checked: (1) $\iota_h$ is a bijection, and (2) the internal flow graph is acyclic. 
		
		As to the former, the top horizontal map in diagram (\ref{eq:iota_h_diagram}) is an isomorphism, and both it and its inverse respect the equivalence relations by construction. Thus, it descends to an isomorphism
		\[
		\left(\underline{m}\setminus (\iota^X_h)^{-1}(\underline{n})\right)\amalg \left(\coprod_{x\in X} O_h^x\right)\cong \left(\underline{n}\setminus\iota_h^X(\underline{m})\right)\amalg\left(\coprod_{x\in X}I_h^x\right).
		\]
		Since the bottom map is simply an isomorphism on the complements of these subsets, it follows that $\iota_h$ is an isomorphism. 
		
		To see (2), suppose, to the contrary, there is an oriented cycle in the flow graph $G$ of $\Gamma\bullet \{\Lambda^x\}_{x\in X}$. We will write this cycle as a chain of vertices and edges 
		\[
		\begin{tikzcd}
			y_0 \arrow[r,"e_1"] &y_1 \arrow[r,"e_2"] & \cdots \arrow[r,"e_k"] &y_k\arrow[r,"e_{k+1}"] & y_0. 
		\end{tikzcd}
		\]
		Identifying the flow graphs $H^x$ of the $\Lambda^x$ with subgraphs of $G$, we can note that, were this cycle contained in any one of them, $\Lambda^x$ would not be acyclic, a contradiction. We can thus choose a sequence $x_0,\ldots x_n\in X$ of vertices of $\Gamma$ which partition the sequence $\{y_i\}$ by the subgraphs $H^x$. But then there is an oriented cycle
		\[
		\begin{tikzcd}
			x_0 \arrow[r] & x_1 \arrow[r] & \cdots \arrow[r] &x_n \arrow[r] & x_0
		\end{tikzcd}
		\]
		in the flow graph of $\Gamma$, a contradiction.
	\end{proof}

	\begin{prop}
		The operation $\bullet$ is associative up to isomorphism. That is 
		\[
		\Gamma \bullet\left\lbrace\Lambda^x\bullet \{\Xi^x_y\}_{y\in Y^x} \right\rbrace_{x\in X}\cong \left(\Gamma \bullet \{\Lambda^x\}_{x\in X}\right)\bullet\{\Xi^x_y\}_{x\in X, y\in Y^x}.
		\]
		Moreover, $\bullet$ is right-unital up to isomorphism, with unit given by the unique (up to isomorphism) collection of one-vertex port graphs forming a labeling of $\Gamma$. 
	\end{prop}
	

\begin{eg}\label{ex:pasting}
{The pasting operation is very natural when represented as boxes.} {Consider the labeled double port graph}

 \begin{center}
\begin{tikzpicture}[x=0.5em,y=0.5em]
\def\xshift{12}

  \Block{\xshift}{0}{$\Lambda^a$}{Main} 

  \HIndraw{Main}{2}{1}{3}
  \HIndraw{Main}{2}{2}{3}

  \HOutdraw{Main}{1}{1}{7}

  \VIndraw{Main}{1}{1}{3}
  \VOutdraw{Main}{3}{1}{11.1}

  \Block{\xshift+4}{-8}{$\Lambda^b$}{Down}

  \VDraw{Main}{3}{2}{Down}{2}{1}{-90}{120}
  \VDraw{Main}{3}{3}{Down}{2}{2}{-90}{70}

  \VOutdraw{Down}{1}{1}{3}
  \HOutdraw{Down}{1}{1}{3}

\end{tikzpicture}
\end{center}
{where the labels are given by the following double port graphs}
\begin{center}
\begin{tikzpicture}[x=0.5em,y=0.5em]

  \def\xA{12}

  \Block{\xA}{0}{}{MainA}
  \Block{\xA+10}{4}{}{UpA}

  \HIndraw{MainA}{2}{1}{3}
  \HIndraw{MainA}{2}{2}{3}

  \HOutdraw{MainA}{2}{2}{13}

  \VIndraw{MainA}{1}{1}{6}
  \VOutdraw{MainA}{2}{1}{3}
  \VOutdraw{MainA}{2}{2}{3}

  \HDraw{MainA}{2}{1}{UpA}{1}{1}{0}{180}

  \VOutdraw{UpA}{1}{1}{7}

  \node at (\xA+5,-8) {$\Lambda_a$};

  \def\xB{38}

  \Block{\xB}{0}{}{DownB}

  \VIndraw{DownB}{2}{1}{3}
  \VIndraw{DownB}{2}{2}{3}

  \VOutdraw{DownB}{1}{1}{3}

  \HOutdraw{DownB}{1}{1}{3}

  \node at (\xB,-8) {$\Lambda_b$};

\end{tikzpicture}
\end{center}
{The pasting $\Gamma \bullet \set{ \Lambda^{a},\Lambda^b}$ gives the double port graph in Example \ref{ex:double-port}. More explicitly, each label is drawn inside the corresponding box and the wires are connected using the labeling data. Then, the pasting is obtained by removing the outer boxes:}
\begin{center}
\begin{tikzpicture}[x=0.5em,y=0.5em]
\def\xshift{12}

  \Block{\xshift}{0}{}{Main}
  \Block{\xshift+10}{4}{}{Up}

  \HIndraw{Main}{2}{1}{3}
  \HIndraw{Main}{2}{2}{3}

  \HOutdraw{Main}{2}{2}{13.2}

  \VIndraw{Main}{1}{1}{8}
  \VOutdraw{Main}{2}{1}{12}

  \HDraw{Main}{2}{1}{Up}{1}{1}{0}{180}

  \draw[very thick]
    ($(Main.south west)+(-1,-1)$)
    rectangle
    ($(Up.north east)+(1,1)$);

  \Block{\xshift+10}{-12}{}{Down}

  \VDraw{Main}{2}{2}{Down}{2}{1}{-90}{90}
  \VDraw{Up}{1}{1}{Down}{2}{2}{-90}{90}

  \VOutdraw{Down}{1}{1}{3}
  \HOutdraw{Down}{1}{1}{3}

  \draw[very thick]
    ($(Down.south west)+(-1,-1)$)
    rectangle
    ($(Down.north east)+(1,1)$);

\end{tikzpicture}
\end{center}
\end{eg}

	\subsection{Functoriality in port labels} 
\label{sec:functoriality port labels}	
	
	We now turn to the various ways in which transformations of port labels can provide functors of double categories.

	\begin{defn}
		Let $\lL,\mM$ be sets of port labels. A \emph{morphism of port labels} $\lL\to \mM$ is a map of sets 
		\[
		\begin{tikzcd}
			f: &[-3em] \lL\arrow[r] & \mM 
		\end{tikzcd}
		\]
		which commutes with the maps $s^v,t^v,s^h,$ and $t^h$. We denote the category of sets of port labels with morphisms of port labels by $\sf{PL}$. 
	\end{defn} 
	
	Given a morphism $f:\lL\to \mM$ of port labels, we construct a double functor 
	\[
	\begin{tikzcd}
		\overline{f}:&[-3em] \sf{DPG}_\lL\arrow[r] & \sf{DPG}_\mM 
	\end{tikzcd}
	\]
	as follows. On objects, horizontal morphisms, and vertical morphisms, $\overline{f}$ is the identity. On squares, given a double port graph $(X,\on{in}_v,\on{out}_v,\on{in}_h,\on{out}_h,\iota_v,\iota_h)$ with $\lL$-labels $\on{lab}:X\to \lL$, we send this to the same double port graph, but now with labels $f\circ \on{lab}:X\to \mM$. It is immediate that this is compatible with identities and composition, and thus defines a double functor. Similarly immediate is the following proposition. {Let $\sf{DCat}$ denote the category of double categories and double functors between them.}
	
	\begin{prop}
		The assignments $\lL\mapsto \sf{DPG}_\lL$ and $(f:\lL\to \mM)\mapsto \overline{f}$ define a 
		functor
		\[
		\begin{tikzcd}
			\sf{DPG}: &[-3em] \sf{PL} \arrow[r] & \sf{DCat}
		\end{tikzcd}
		\]
		to the category of double categories. 		
\end{prop}

For any set ${\lL}$ of port labels, the set $\on{Sq}(\mathsf{DPG}_{{\lL}})$ 
of squares of $\mathsf{DPG}_{{\lL}}$ 
is itself a set of port labels, equipped with the horizontal and vertical source and target maps of $\mathsf{DPG}_{{\lL}}$. 
{Then,} the canonical map 
\[
\begin{tikzcd}
{\eta} : &[-3em] \lL \arrow[r] & \on{Sq}(\mathsf{DPG}_\lL)
\end{tikzcd}
\]
sending a label $L$ to the double port graph {$\Gamma_L$} with a single $L$-labeled vertex 
defines a morphism of port labels. Consequently, there is a canonical double functor
\[
\begin{tikzcd}
{\overline{\eta}} : &[-3em] \mathsf{DPG}_{{\lL}} \arrow[r] & \mathsf{DPG}_{\on{Sq}(\mathsf{DPG}_{{\lL}})}.
\end{tikzcd}
\]
{Next, we describe a construction that yields a double functor in the opposite direction.}

Given a square in $\sf{DPG}_{\on{Sq}(\sf{DPG}_\lL)}$, that is, a double port graph $\Gamma$ whose labels are represented by the $\lL$-labeled double port graphs $\{\Lambda^x\}_{x\in X}$, we send this square to the double port graph 
$\Gamma \bullet \{\Lambda^x\}_{x\in X}$,
whose $\lL$-labels are given by the labels of the $\Lambda^x$. This construction defines a {double} functor 
\[
\begin{tikzcd}
\on{paste}_{\lL}:&[-3em]\sf{DPG}_{\on{Sq}(\sf{DPG}_\lL)} \arrow[r] & \sf{DPG}_{\lL}.
\end{tikzcd}
\] 
As a consequence, given a set $\mM$ of port labels and a morphism of port labels 
\[
\begin{tikzcd}
f: &[-3em] \mM \arrow[r] & \on{Sq}(\sf{DPG}_\lL),
\end{tikzcd}
\]
we obtain a composite double functor 
\begin{equation}\label{eq:f lower star}
\begin{tikzcd}
{\kappa_{\mM,\lL}:}&[-3em]	\sf{DPG}_\mM \arrow[r,"\overline{f}"] & \sf{DPG}_{\on{Sq}(\sf{DPG}_\lL)} \arrow[r,"\on{paste}_\lL"] & \sf{DPG}_\lL.
\end{tikzcd}	
\end{equation}

\section{{Adaptive instruments}}
\label{sec:double cat of inst}

In this section, we introduce the double category of adaptive instruments.

\subsection{Instruments}
\label{sec:instruments}

For background on quantum theory, we refer to \cite{watrous2018theory,heunen2019categories}. 
A quantum system is described by a Hilbert space, which is typically taken to be finite-dimensional in quantum computing. We will have this restriction throughout the paper.
Hilbert spaces and linear maps between them can be assembled into a category, denoted by $\sf{Hilb}$ \cite{heunen2019categories}. The tensor product $\otimes$ operation of Hilbert spaces endows this category with a symmetric monoidal structure.

For a Hilbert space $V$, we will write $L(V)$ to denote the space of linear maps (operators) on $V$. For a linear operator $A$, we will write $A^\dagger$ for its adjoint. The following operators are important in quantum theory: $A$ is called
\begin{itemize}
\item \emph{Hermitian} if $A=A^\dagger$,
\item \emph{positive semi-definite} if it can be written as $A=B^\dagger B$ for some operator $B$,
\item \emph{projection} if {$A$ is Hermitian and} $A^2=A$,
\item \emph{unitary} if its inverse is $A^\dagger$. 
\end{itemize}
Positive-semidefinite operators are important in both representing quantum states and quantum measurements. A positive-semidefinite operators whose trace is $1$ is called a \emph{density operator}, also called \emph{a quantum state}. A \emph{quantum measurement} with {finite} output set ${Y}$ is given by a function
\[
\begin{tikzcd}
		\Pi: &[-3em] {Y} \arrow[r] & \on{Pos}(V)
\end{tikzcd}
\]
such that $\sum_{a\in {Y}}\Pi(a)=\one_V$. The quantum measurement is said to be \emph{projective} if the function lands in the set of projectors. A quantum state $\rho$ and a quantum measurement {$\Pi$} gives a probability distribution $p$ defined by the \emph{Born rule}:
\[
p(a) = \on{Tr}(\rho \Pi(a)).
\]

A linear map $\phi:L(V)\to L(W)$ is called \emph{positive} if it maps positive semi-definite operators to positive semi-definite operators. The map $\phi$ is called \emph{completely positive} if {$\phi\otimes \one_{L(U)}$} is positive for any Hilbert space {$U$}. A completely positive linear map is called a \emph{channel} if it is trace-preserving, {i.e., $\on{Tr}(\phi(A))=\on{Tr}(A)$}. We can assemble channels into a category. We will write $\on{CP}(V,W)$ and $\on{C}(V,W)$ for completely positive linear maps and channels, respectively.

\begin{defn}\label{def:cat-channels}
The \emph{category of channels}, denoted by $\mathsf{Chan}$, is defined as follows: 
its objects are Hilbert spaces, and its morphisms are channels between them.
\end{defn}

Quantum measurements and more general quantum operations can be formulated using instruments.
An \emph{instrument} with {output} set ${Y}$ is a function  
\[
\begin{tikzcd}
		\Phi: &[-3em] {Y} \arrow[r] & \on{CP}(V,W)
	\end{tikzcd}
\]
such that $\sum_{a\in {Y}} \Phi^a$ is a channel. Here we adopt the notation that $\Phi^a=\Phi(a)$. Two important examples of instruments are given by (1) unitary transformations, and (2) quantum measurements. In the former case, the output set is a singleton and the channel is given by conjugating with the unitary. In the latter, the output set of the instrument is the same as the output set of the quantum measurement {and t}he instrument is given by $\Phi^a(-)=\Pi(a)(-)\Pi(a)$. The Born rule generalizes to this case to give a probability distribution associated to a quantum state $\rho$ and an instrument {$\Phi$:}
\[
p(a) = \on{Tr}(\Phi^a(\rho)).
\]

{
In categorical quantum mechanics, dagger compact categories are used to axiomatize the structural properties of the category $\sf{Hilb}$ of finite-dimensional Hilbert spaces; see, for example, \cite{heunen2019categories}. The CPM construction \cite{selinger2007dagger} provides a general procedure for passing from a dagger compact category of pure processes to a category of completely positive maps. In the case of $\sf{Hilb}$, restricting to trace-preserving maps yields the category $\sf{Chan}$ of quantum channels.
}

\subsection{The double category}	
	\label{sec:double cat}

We now introduce the main objects of this section, namely \emph{adaptive instruments}, and construct the associated double category. As we will see in the following sections, adaptive instruments provide models for quantum computation~\cite{okay2024classical,okay2025polyhedral}.

\begin{defn}\label{def:double cat inst}
The \emph{double category of instruments}, denoted by $\mathsf{Inst}$, is a $1$-object double category whose
\emph{horizontal morphisms} are sets $\underline n$, $n\geq 0$, and whose
\emph{vertical morphisms} are Hilbert spaces $\CC[\underline{n}]$, $n\geq 0$. 
A \emph{square} in $\mathsf{Inst}$, bounded by the sets $X=\underline{k}$, $Y=\underline{\ell}$ and the Hilbert spaces $V=\CC[\underline{m}]$, $W=\CC[\underline{n}]$, is of the form
\[
\begin{tikzcd}
	\ast \arrow[r, "{X}", ""'{name=U}] \arrow[d, "V"'] 
	& \ast \arrow[d, "W"] \\
	\ast \arrow[r, "{Y}"', ""{name=T}] 
	& \ast \arrow[r, phantom, from=U, to=T, "\Phi"{description}]
\end{tikzcd}
\]
and consists of a map
\[
\begin{tikzcd}
	\Phi : {X} \times {Y} \arrow[r] & \operatorname{CP}(V, W).
\end{tikzcd}
\] 
For each pair $(a,b) \in {X} \times {Y}$, we write $\Phi_a^b$ for the corresponding element of $\operatorname{CP}(V, W)$. {We require that, for each $a\in X$, the function $\Phi_a:Y\to \on{CP}(V,W)$ is an instrument, meaning that it satisfies $\sum_{b\in Y} \Phi_a^b \in \on{C}(V,W)$.} 
Here ${X}$ and ${Y}$ are interpreted as the input and output sets, respectively.    
The {squares} of this double category, that is, the functions $\Phi$ as defined above, are called \emph{adaptive instruments}.
\end{defn}

{An equivalent way to encode an adaptive instrument is as the linear map}
\begin{equation}\label{eq:CPTP version of adap int}
	\begin{tikzcd}[row sep=.5ex]
		\widehat{\Phi}:&[-3em] L(\CC[X]\otimes V) \arrow[r] & L(\CC[Y]\otimes W)\\
		& \ket{a}\bra{a'}\otimes A \arrow[r,mapsto] & \delta_{a,a'} \sum_{b\in Y} \ket{b}\bra{b}\otimes \Phi_a^b(A). 
	\end{tikzcd}
\end{equation}
{This map is a channel. Conversely, every channel of this form determines an adaptive instrument by reading off the maps $\Phi_a^b$.}

{
\begin{rmk}
In this definition, we restrict to $\sf{Fin}_\NN$ and to the skeletal
subcategory of $\mathsf{Hilb}$, namely, the full subcategory
$\sf{Hilb}_\NN$ whose objects are $\mathbb{C}[\underline{n}]$ for
$n \ge 0$. The former is a strict monoidal category under the Cartesian
product
$\underline{n} \times \underline{m}=\underline{nm}$, where the identification is given by the lexicographic bijection, while the latter is a
strict monoidal category under the tensor product
$\CC[\underline{n}] \otimes \CC[\underline{m}] = \CC[\underline{nm}]$.
\end{rmk}
}

{For the rest of this section will show {Definition \ref{def:double cat inst}} actually yields a well-defined double category.}	  
		{Let us begin with defining the compositions.}
	The horizontal composition of the diagrams 
	\[
	\begin{tikzcd}
		\ast \arrow[r,"{X}_1",""'{name=U}]\arrow[d,"V"'] & \ast\arrow[d,"W"]\\
		\ast \arrow[r,"{Y}_1"',""{name=T}] & \ast  \arrow[r,phantom,from=U,to=T,"\Phi"{description}]
	\end{tikzcd}
	\quad \text{and} \quad \begin{tikzcd}
		\ast \arrow[r,"{X}_2",""'{name=U}]\arrow[d,"W"'] & \ast\arrow[d,"U"]\\
		\ast \arrow[r,"{Y}_2"',""{name=T}] & \ast  \arrow[r,phantom,from=U,to=T,"\Psi"{description}]
	\end{tikzcd}
	\]
	is the square $\Psi\circ \Phi$ defined by 
	\[
	(\Psi\circ \Phi)_{(a_1,a_2)}^{(b_1,b_2)}=\Psi_{a_2}^{b_2}\circ \Phi_{a_1}^{b_1} .
	\]
	Since the composition of completely positive operators is completely positive, this results in a completely positive operator. Moreover, 
	\[
	\begin{aligned}
		\sum_{(b_1,b_2)} (\Psi\circ \Phi)_{(a_1,a_2)}^{(b_1,b_2)} & =\sum_{(b_1,b_2)}\Psi_{a_2}^{b_2}\circ \Phi_{a_1}^{b_1} \\
		& =\sum_{b_1}\sum_{b_2}\Psi_{a_2}^{b_2}\circ \Phi_{a_1}^{b_1}\\
		&= \left(\sum_{b_2}\Psi_{a_2}^{b_2}\right) \circ \left(\sum_{b_1} \Phi_{a_1}^{b_1}\right)
	\end{aligned}
	\]
	is a composite of channels, and is thus itself a channel.

	The horizontal identity morphism is defined by the square
	\[
	\begin{tikzcd}[row sep=.5ex]
		\on{Id}^h_{V} :&[-3em] \ast \times \ast \arrow[r] & {\on{CP}(V,V)}\\
		& (\ast,\ast) \arrow[r] & {\one_V} .
	\end{tikzcd}
	\]		
We have $(\Phi \circ \on{Id}^h_{V})=(\on{Id}^h_{V} \circ \Phi)=\Phi$.
	
	The {vertical} composition of 
	\[
	\begin{tikzcd}
		\ast \arrow[r,"{X}",""'{name=U}]\arrow[d,"V_1"'] & \ast\arrow[d,"W_1"]\\
		\ast \arrow[r,"{Y}"',""{name=T}] & \ast  \arrow[r,phantom,from=U,to=T,"\Phi"{description}]
	\end{tikzcd}
	\]
	and 
	\[
	\begin{tikzcd}
		\ast \arrow[r,"{Y}",""'{name=U}]\arrow[d,"V_2"'] & \ast\arrow[d,"W_2"]\\
		\ast \arrow[r,"{Z}"',""{name=T}] & \ast  \arrow[r,phantom,from=U,to=T,"\Psi"{description}]
	\end{tikzcd}
	\]
	is given by the map  
	\[
	\begin{tikzcd}[row sep=.5ex]
		(\Psi\bullet \Phi):&[-3em] {X}\times {{Z}} \arrow[r] & \on{CP}(V_1\otimes V_2,W_1\otimes W_2)\\
		& (a,c) \arrow[r,mapsto] & \sum_b \Phi_a^b\otimes \Psi_b^c. \\
	\end{tikzcd}
	\]
{Here, we use the identification $L(V_1)\otimes L(V_2)\cong L(V_1\otimes V_2)$ obtained}	
via the Kroenecker product. As above, this is clearly associative, since the Kroenecker product is, thus, we need to show that it is unital and yields completely positive operators which satisfy the desired channel condition.
{Complete positivity follows from the facts that the tensor product of two completely positive maps is still completely positive, and sum of completely positive maps remains to be so. Trace preservation of $\sum_{c} (\Psi\bullet \Phi)_a^c$ follows from the properties of trace.}
	


	As to unitality: The vertical identity morphism {is defined by} the square  
	\[
	\begin{tikzcd}[row sep=.5ex]
		\on{Id}^v_{{X}} :&[-3em] {X} \times {X} \arrow[r] & {\on{CP}(\CC,\CC)}\\
		& (a,b) \arrow[r] & \delta_{a,b}{\one_{\CC}} .
	\end{tikzcd}
	\]
	We then note that 
	\[
	\left(\Phi \bullet  \on{Id}_{X}^v\right)_a^c= \sum_{b\in {X}} \delta_{a,b} {\one_{\CC}} \otimes \Phi_b^c =\Phi_a^c 
	\]
	and 
	\[
	\left(\on{Id}_{Y}^v \bullet \Phi\right)_a^c =\sum_{b\in {Y}} \Phi_a^b \otimes (\delta_{b,c}{\one_{\CC}})= \Phi_a^c  
	\]
	showing that this is unital. 
	
 Furthermore,	
	it is immediate that 
	\[
	\on{Id}_\ast^v =\on{Id}_\CC^h.
	\]

	The final check we need to perform is the interchange law. 
	 
\begin{lem}\label{lem:interchange law inst}
Given the following diagram
\[
\begin{tikzcd}[row sep=4em,column sep=4em]
	\ast
		\arrow[r,"{X}_1"]
		\arrow[d,"V_1"']
		\arrow[dr,phantom,"\Phi_{1,1}" description]
	&
	\ast
		\arrow[r,"{{X}_2}"]
		\arrow[d,"W_1"]
		\arrow[dr,phantom,"\Phi_{1,2}" description]
	&
	\ast\arrow[d,"U_1"]
	\\
	\ast
		\arrow[r,"{Y}_1"]
		\arrow[d,"V_2"']
		\arrow[dr,phantom,"\Phi_{2,1}" description]
	&
	\ast
		\arrow[r,"{{Y}_2}"]
		\arrow[d,"W_2"]
		\arrow[dr,phantom,"\Phi_{2,2}" description]
	&
	\ast\arrow[d,"U_2"]
	\\
	\ast\arrow[r,"{{Z}_1}"]
	&
	\ast\arrow[r,"{Z}_2"]
	&
	\ast
\end{tikzcd}
\] 
 {the interchange law holds:}
	\[
	\left((\Phi_{2,2}\circ \Phi_{2,1})\bullet(\Phi_{1,2}\circ \Phi_{1,1})\right)= \left((\Phi_{2,2}\bullet \Phi_{1,2})\circ (\Phi_{2,1}\bullet \Phi_{1,1})\right).
	\] 
\end{lem}	
	\begin{proof}
	 We compute 
	\[
	\begin{aligned}
		\left((\Phi_{2,2}\circ \Phi_{2,1})\bullet(\Phi_{1,2}\circ \Phi_{1,1})\right)_{(a_1,a_2)}^{(c_1,c_2)}& =\sum_{(b_1,b_2)}(\Phi_{1,2}\circ \Phi_{1,1})_{(a_1,a_2)}^{(b_1,b_2)}\otimes (\Phi_{2,2}\circ \Phi_{2,1})_{(b_1,b_2)}^{(c_1,c_2)}\\
		&= \sum_{(b_1,b_2)}\left((\Phi_{1,2})_{a_2}^{b_2}\circ (\Phi_{1,1})_{a_1}^{b_1}\right)\otimes \left((\Phi_{2,2})_{b_2}^{c_2}\circ (\Phi_{2,1})_{b_1}^{c_1}\right)\\
		&= \sum_{b_1}\sum_{b_2} \left((\Phi_{1,2})_{a_2}^{b_2}\otimes (\Phi_{2,2})_{b_2}^{c_2} \right)\circ \left((\Phi_{1,1})_{a_1}^{b_1}\otimes (\Phi_{2,1})_{b_1}^{c_1}\right)\\
		&= \left(\sum_{b_2}(\Phi_{1,2})_{a_2}^{b_2}\otimes (\Phi_{2,2})_{b_2}^{c_2} \right)\circ \left(\sum_{b_1}(\Phi_{1,1})_{a_1}^{b_1}\otimes (\Phi_{2,1})_{b_1}^{c_1}\right)\\
		&= \left(
		\Phi_{2,2}\bullet\Phi_{1,2}
		\right)_{a_2}^{c_2}\circ \left(
		\Phi_{2,1}\bullet \Phi_{1,1}
		\right)_{a_1}^{c_1}\\
		&= \left((\Phi_{2,2}\bullet \Phi_{1,2})\circ (\Phi_{2,1}\bullet \Phi_{1,1})\right)_{(a_1,a_2)}^{(c_1,c_2)}
	\end{aligned} 
	\]
	Here, the only non-trivial manipulations we have used are (1) the interchange law for tensor products and composition, and (2), the bilinearity of composition with respect to sums. {Thus, the interchange law holds.} 
\end{proof}

So we see that Definition \ref{def:double cat inst} gives a well-defined 1-object  double category.

\subsection{{The horizontal and vertical {monoidal} categories}} 
	\label{sec:hor and ver cat}

Now we consider the horizontal and vertical monoidal categories associated with the $1$-object double category of instruments ({see} Definition \ref{def:h v monoidal} {in Section \ref{sec:double category}}).

The horizontal monoidal category $\sf{H}(\sf{Inst})$ consists of objects given by Hilbert spaces {of the form $\CC[\underline{n}]$, $n \geq 0$}.  
Its morphisms are obtained from squares by setting both the input ${X}$ and output ${Y}$ to be singletons.  
A square of the form
\[
\begin{tikzcd}
	\ast \arrow[r,"\ast",""'{name=U}]\arrow[d,"V"'] & \ast \arrow[d,"W"]\\
	\ast \arrow[r,"\ast"',""{name=T}] & \ast  \arrow[r,phantom,from=U,to=T,"\Phi"{description}]
\end{tikzcd}
\]
corresponds to an instrument $\ast \times \ast \to \on{CP}(V,W)$, equivalently to a quantum channel from $V$ to $W$.  
The horizontal composition in $\sf{H}(\sf{Inst})$ coincides with the composition of channels.

\begin{prop}\label{pro:horizontal}
The horizontal monoidal category $\sf{H}(\sf{Inst})$ can be identified with
the strict monoidal category $({\sf{Chan}_\NN},\otimes,\CC)$, whose
underlying category is the full subcategory of $\sf{Chan}$ on the objects
$\CC[\underline{n}]$, $n\geq 0$.
\end{prop}

\begin{defn}\label{def:distribution monad}
For a set ${Y}$, the set of probability distributions on ${Y}$ is defined as
\[
D({Y}) := \left\lbrace\, p : {Y} \to \RR_{\geq 0}
\colon p \text{ finitely supported and } \sum_{a \in {Y}} p(a) = 1 \,\right\rbrace.
\]
This defines a functor $D : \sf{Set} \to \sf{Set}$, called the \emph{distribution monad} \cite{jacobs2010convexity}.
\end{defn}
  
The unit map $\delta : {Y} \to D({Y})$ is given by the Dirac (delta) distribution at each element, while the multiplication map $\mu : D^2({Y}) \to D({Y})$ is induced by the multiplication in the semiring~$\RR_{\geq 0}$. 
The associated Kleisli category {$\sf{Stoch}$} has sets as objects and morphisms of the form ${X} \to D({Y})$. 
It is symmetric monoidal: letting $m : D({Y}) \times D({Y}') \to D({Y} \times {Y}')$ denote the map sending a pair of distributions $(p,q)$ to their product,
\[
p \cdot q(a,b) = p(a)\, q(b),
\]
the tensor product $\odot$ in {$\sf{Stoch}$} is given on objects by ${Y} \odot {Y}' = {Y} \times {Y}'$, and on morphisms by
\[
p \odot q = m \circ (p \times q)
\]
for $p : {X} \to D({Y})$ and $q : {X}' \to D({Y}')$.

Now, turning to the vertical monoidal category $\sf{V}(\sf{Inst})$, we observe that its objects are given by {$\underline{n}$, $n\geq 0$}. A morphism in this category is represented by a square of the form
\[
\begin{tikzcd}
	\ast \arrow[r,"{X}",""'{name=U}]\arrow[d,"\CC"'] & \ast \arrow[d,"\CC"]\\
	\ast \arrow[r,"{Y}"',""{name=T}] & \ast \arrow[r,phantom,from=U,to=T,"\Phi"{description}]
\end{tikzcd}
\]
that is, by an instrument of the form ${X} \times {Y} \to {\on{CP}}(\CC,\CC)$.  
Since ${\on{CP}}(\CC,\CC) \cong \RR_{\geq 0}$, such instruments correspond precisely to Kleisli morphisms of the form ${X} \to D({Y})$.

\begin{prop}\label{pro:vertical}
The vertical monoidal category $\sf{V}(\sf{Inst})$ can be identified with
the strict monoidal category $({\sf{Stoch}_\NN},\odot,\ast)$, whose
underlying category is the full subcategory of $\sf{Stoch}$ on the objects
$\underline{n}$, for $n\geq 0$.
\end{prop}

\section{{Quantum computation with qubits}}	
\label{sec:quantum computation}

{In this section, we apply our double categorical framework of double port graphs and adaptive instruments to classical and quantum computation.} 
We introduce the double port graphs for Boolean circuits and quantum circuits, and construct the corresponding double functors into the double category of instruments, 
which endow these abstract port graphs with their operational meaning. In effect these constructions amount to describing universal gate sets for Boolean and quantum circuits, which is well-known \cite{kitaev2002classical}.

\subsection{Qubit instruments} 
  
A \emph{qubit system} is represented by the Hilbert space $\CC^2$.  
Taking an $n$-fold tensor product yields the Hilbert space $(\CC^2)^{\otimes n}$ of an \emph{$n$-qubit system}.  
{We may identify $(\CC^2)^{\otimes n}$ with the group algebra $\CC[\ZZ_2^n]$, which in turn can be identified with $\CC[\underline{2^n}]$ via the bijection $\ZZ_2^n \to \underline{2^n}$ sending each bit string $(x_1,\dots,x_n)$ to the natural number $\sum_{i=1}^n 2^{n-i}x_i$.}

  {
\begin{defn}\label{def:qubit inst}
The \emph{double category of qubit instruments}, denoted by
$\mathsf{Qbit}$, is a $1$-object double category whose \emph{vertical
morphisms} are the Hilbert spaces $(\CC^2)^{\otimes n}$, $n\geq 0$, and whose
\emph{horizontal morphisms} are the sets $\ZZ_2^n$, $n\geq 0$. Its
squares are instruments of the form
\[
\left(
\begin{smallmatrix}
&\ZZ_2^k&\\
(\CC^2)^{\otimes m}&&(\CC^2)^{\otimes n}\\
&\ZZ_2^\ell&
\end{smallmatrix}
\right).
\]
\end{defn} 

For convenience, we will simply write
\[
\left(
\begin{smallmatrix}
&k&\\
m&&n\\
&\ell&
\end{smallmatrix}
\right)
\]
for such a square and thus regard both the vertical and horizontal morphisms
as natural numbers. Such a square takes $m$ qubits and $k$ bits as input and
produces $n$ qubits and $\ell$ bits as output.

\begin{rmk}
The strictness of $\mathsf{Qbit}$ is obtained by using the identification
$\ZZ_2^n\times \ZZ_2^{n'}
=
\ZZ_2^{n+n'}$
and the corresponding Hilbert space identification
$
\CC[\ZZ_2^n]\otimes \CC[\ZZ_2^{n'}] =
\CC[\ZZ_2^{n+n'}]$.
\end{rmk}
}

\begin{defn}\label{def:Bool-Qubit}
We define the following categories:
\begin{itemize} 
  \item $\mathsf{Bool}$ denotes the strict monoidal category whose underlying category is the full subcategory of $\sf{Fin}$ on the objects $\ZZ_2^n$, where $n\geq 0$.
  
  \item $\mathsf{QChan}$ denotes the strict monoidal category whose underlying category is the full subcategory of $\mathsf{Chan}$ on the qubit Hilbert spaces $(\CC^2)^{\otimes n}$, where $n\geq 0$.
\end{itemize}
\end{defn}

As a consequence of Propositions \ref{pro:horizontal} and \ref{pro:vertical}, the horizontal and vertical monoidal categories of $\sf{QBit}$ can be identified with $\sf{QChan}$ and $\sf{Bool}_D$, respectively.

\subsection{Boolean circuits}
\label{sec:boolean circuits}

A function $f : \ZZ_2^m \to \ZZ_2^\ell$ is called a \emph{Boolean map (or function)}. Such maps form the fundamental building blocks of computation over classical bits \cite{lafont2003towards}. In what follows, we construct a category of port graphs labeled by the primitive operations used to realize these maps.

\begin{defn}
\label{def:Boolean labels}
The set $\bB$ of  
\emph{Boolean labels} consists of the \emph{$\bf{1}$-state}, \textit{delete}, \textit{XOR}, \textit{OR}, and \textit{copy} labels:
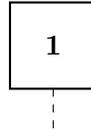
\begin{center}
 		\begin{tikzpicture}[x=0.5em,y=0.5em]
 			\Block{0}{0}{$\mathbf{1}$}{F}
 			\VOutdraw{F}{1}{1}{3}
 		\end{tikzpicture}\hspace{3em}
 		\raisebox{1.5em}
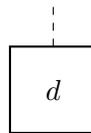
{
 	 		\begin{tikzpicture}[x=0.5em,y=0.5em]
 				\Block{0}{0}{${d}$}{F}
 				\VIndraw{F}{1}{1}{3}
 			\end{tikzpicture}
 		}
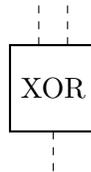
\hspace{3em}
 	 		\begin{tikzpicture}[x=0.5em,y=0.5em]
 			\Block{0}{0}{$\on{XOR}$}{XOR}
 			\VOutdraw{XOR}{1}{1}{3}
 			\VIndraw{XOR}{2}{1}{3}
 			\VIndraw{XOR}{2}{2}{3}
 		\end{tikzpicture}
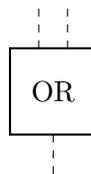
\hspace{3em}
 		\begin{tikzpicture}[x=0.5em,y=0.5em]
 			\Block{0}{0}{$\on{OR}$}{XOR}
 			\VOutdraw{XOR}{1}{1}{3}
 			\VIndraw{XOR}{2}{1}{3}
 			\VIndraw{XOR}{2}{2}{3}
 		\end{tikzpicture}
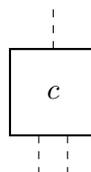
\hspace{3em}
 		\begin{tikzpicture}[x=0.5em,y=0.5em]
 			\Block{0}{0}{$c$}{C}
 			\VIndraw{C}{1}{1}{3}
 			\VOutdraw{C}{2}{1}{3}
 			\VOutdraw{C}{2}{2}{3}
 		\end{tikzpicture}
\end{center}
The category of \emph{Boolean-labeled port graphs} is denoted by $\sf{PG}_{\bB}$.  
\end{defn}

Composition in this category is in the vertical direction. Treating this as a double category with no horizontal composition we can realize a double functor into instruments.

\begin{defn}
The double functor 
\[
\begin{tikzcd}
\phi_\bB: &[-3em]	
{\sf{PG}_{\bB}} \arrow[r]   & 
\sf{QBit}
\end{tikzcd}
\]
is determined by sending the labels $\bf{1}$, $d$, $\on{XOR}$, $\on{OR}$, and $c$ to the corresponding instruments, as follows:
 \begin{itemize}
 	\item The \textit{$\mathbf{1}$-state instrument}, for $s\in {\ZZ_2}$ and $\chi\in \CC$, 
 	\[
 	\Phi_{\mathbf{1}}^s(\chi):= \delta_{s,1} \, \chi.
 	\]
 	 
 	\item The \textit{delete instrument}, for $s\in {\ZZ_2}$ and $\chi\in \CC$,
 	\[
 	({\Phi_d})_s(\chi):=\chi .
 	\]

 	\item The \textit{XOR instrument}, for $((s,r),t)\in {\ZZ_2^2\times \ZZ_2}$ and $\chi \, \in \CC$,  
 	\[
(\Phi_{\on{XOR}})_{s,r}^t(\chi):=\delta_{t,s\oplus r} \, \chi
 	\] 
 	where $\oplus$ denotes the mod-$2$ sum.
 	 	
 	{
 	\item The \textit{OR instrument}, for $((s,r),t)\in {\ZZ_2^2\times \ZZ_2}$ and $\chi \, \in \CC$,  
 	\[
(\Phi_{\on{OR}})_{s,r}^t(\chi):=\delta_{t,s\lor
 r} \, \chi
 	\] 
 	where $\lor$ denotes the logical OR operation.
 	} 
 	
 	\item The \textit{copy instrument}, for $(r,(s,t))\in {\ZZ_2\times \ZZ_2^2}$ and $\chi\in \CC$,
 	\[
 	({\Phi_{c}})_r^{s,t}(\chi):=\delta_{r,s}\delta_{r,t} \, \chi .
 	\]
 \end{itemize} 
\end{defn}

{
\begin{prop}\label{pro:image phi B is bool}
The image of $\phi_\bB$ can be identified {with} the category $\sf{Bool}$ of Boolean functions.
\end{prop}
}

We also consider the subset ${\bB_{\oplus}} = \bB \setminus \set{\on{OR}}$ of \emph{affine Boolean labels}, and the corresponding double functor from the category of \emph{affine Boolean-labeled double port graphs}, obtained by restricting the label morphism above to this label set:
\[
\begin{tikzcd}
\phi_{{\bB_{\oplus}}}: &[-3em]	 {\sf{PG}_{{\bB_{\oplus}}} }\arrow[r]  & \sf{QBit}.
	\end{tikzcd}
\]

\begin{defn}\label{def:lin}
The \emph{category of affine Boolean maps}, denoted by $\sf{Aff}$, is the subcategory of $\sf{Bool}$ with the same objects and whose morphisms are the affine Boolean maps, i.e., maps whose coordinate functions are of the form
\[
f(x)=a+\sum_{i=1}^n b_i x_i
\]
for some $a,b_i\in\ZZ_2$.
\end{defn}
 
\begin{prop}\label{pro:image phi B is bool}
The image of $\phi_{\bB_\oplus}$ can be identified by the category $\sf{Aff}$ of affine Boolean functions.
\end{prop}

{In later sections, we will use the following labeled port graphs, referred to as the \emph{NOT}, the \emph{$\mathbf{0}$-state}, and the \emph{AND} {labels}, respectively:}
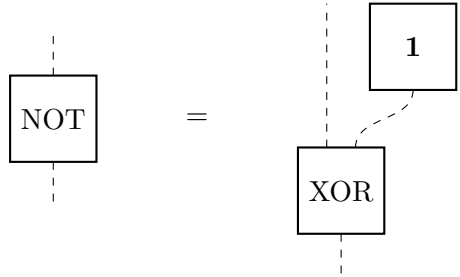
\begin{center}
\def\thirdscale{0.8}
\begin{tabular}{
    @{}c
    @{\hspace{5em}}c
    @{\hspace{5em}}c@{}
}
\begin{tikzpicture}[
    x=0.5em,
    y=0.5em,
    baseline=(current bounding box.center)
]
    \Block{5}{5}{$\mathbf{1}$}{RBT}
    \Block{0}{-5}{$\on{XOR}$}{RBD}

    \VDraw{RBT}{1}{1}{RBD}{2}{2}{-90}{90}
    \VIndraw{RBD}{2}{1}{10}
    \VOutdraw{RBD}{1}{1}{3}
\end{tikzpicture}
&
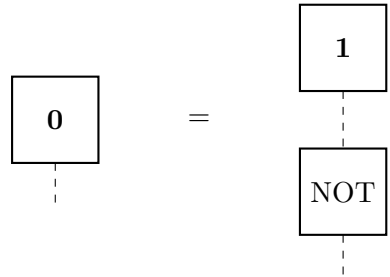
\begin{tikzpicture}[
    x=0.5em,
    y=0.5em,
    baseline=(current bounding box.center)
]
    \Block{0}{5}{$\mathbf{1}$}{RBT}
    \Block{0}{-5}{$\on{NOT}$}{RBD}

    \VDraw{RBT}{1}{1}{RBD}{1}{1}{-90}{90}
    \VOutdraw{RBD}{1}{1}{3}
\end{tikzpicture}
&
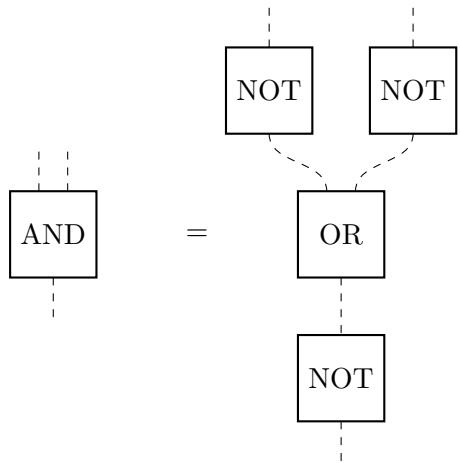
\begin{tikzpicture}[
    x=0.5em,
    y=0.5em,
    scale=\thirdscale,
    transform shape,
    baseline=(current bounding box.center)
]
    \Block{-5}{10}{$\on{NOT}$}{RTL}
    \Block{5}{10}{$\on{NOT}$}{RTR}
    \Block{0}{0}{$\on{OR}$}{RM}
    \Block{0}{-10}{$\on{NOT}$}{RD}

    \VDraw{RTL}{1}{1}{RM}{2}{1}{-90}{90}
    \VDraw{RTR}{1}{1}{RM}{2}{2}{-90}{90}
    \VDraw{RM}{1}{1}{RD}{1}{1}{-90}{90}

    \VIndraw{RTL}{1}{1}{3}
    \VIndraw{RTR}{1}{1}{3}
    \VOutdraw{RD}{1}{1}{3}
\end{tikzpicture}
\end{tabular}
\end{center}

\subsection{Quantum circuits}

In the circuit model \cite{nielsen2010quantum} a quantum computation is represented by a unitary operator
\begin{equation}\label{eq:unitary}
\begin{tikzcd}
		U: &[-3em] (\CC^2)^{\otimes n} \arrow[r] & (\CC^2)^{\otimes n}
	\end{tikzcd}
\end{equation}
followed by a quantum measurement. Typically, the measurement is projective and the projection operators are usually taken to be those that project onto the canonical basis vectors. That is, writing 
\[
\ket{0}:=\begin{pmatrix}
1\\
0
\end{pmatrix}
\;\;\;\;
\text{ and }
\;\;\;\;
\ket{1}:= \begin{pmatrix}
0\\
1
\end{pmatrix},
\]
and writing $\ket{x_1 x_2\dots x_n}:=\ket{x_1}\otimes \ket{x_2}\otimes \dots\otimes \ket{x_n}$, where $x_i\in \set{0,1}$, for the tensor product of these canonical basis vectors, the projectors of the measurement are given by the outer products 
\[
\Pi^x = \ket{x_1 x_2\dots x_n}\bra{x_1 x_2\dots x_n}.
\]
The outer product is understood as the usual matrix multiplication of a column vector with a row vector. 
The readout of a quantum circuit is described by a map
\[
\begin{tikzcd}
p: &[-3em] \ZZ_2^n \arrow[r] & D(\ZZ_2^n).
\end{tikzcd}
\]
{Denoting by $p_r$ the probability distribution associated with an input $r \in \ZZ_2^n$ and by $p_r^s$ the probability of observing the output $s$, we have
\[
p_r^s = \on{Tr}\!\big(\Pi^s\, U\, \Pi^r\, U^\dagger\big).
\]
Typically, the default choice of input is the all-zero string $r = (0,0,\dots,0)$.
}

 
\begin{defn}
\label{def:QC labels}
The set $\cC$ of \emph{QC labels} consists of the \emph{$Z$-preparation}, \emph{Hadamard}, \emph{$T$-gate}, \emph{entangling}, and \emph{(destructive) $Z$-measurement} labels:
\begin{center}
	\begin{tikzpicture}[x=0.5em,y=0.5em]
		\Block{0}{0}{${N_Z}$}{N}
		\VIndraw{N}{1}{1}{3}
		\HOutdraw{N}{1}{1}{3}
	\end{tikzpicture}
	\hspace{2em}
	\begin{tikzpicture}[x=0.5em,y=0.5em]
		\Block{0}{0}{$H$}{H}
		\HIndraw{H}{1}{1}{3}
		\HOutdraw{H}{1}{1}{3}
	\end{tikzpicture}
	\hspace{2em}
	\begin{tikzpicture}[x=0.5em,y=0.5em]
		\Block{0}{0}{$T$}{T}
		\HIndraw{T}{1}{1}{3}
		\HOutdraw{T}{1}{1}{3}
	\end{tikzpicture}
	\hspace{2em}
	\begin{tikzpicture}[x=0.5em,y=0.5em]
		\Block{0}{0}{${E}$}{E}
		\HIndraw{E}{2}{1}{3}
		\HIndraw{E}{2}{2}{3}
		\HOutdraw{E}{2}{1}{3}
		\HOutdraw{E}{2}{2}{3}
	\end{tikzpicture}
	\hspace{2em}
	\raisebox{-1.5em}{
	\begin{tikzpicture}[x=0.5em,y=0.5em]
		\Block{0}{0}{$M_Z$}{M}
		\HIndraw{M}{1}{1}{3}
		\VOutdraw{M}{1}{1}{3}
	\end{tikzpicture}
	}
\end{center}
The double category of \emph{QC-labeled double port graphs} is denoted by $\sf{DPG}_{\cC}$.
\end{defn}

{To define a double functor into the double category of instruments,} we recall some basic quantum operations.  
The Pauli matrices are given by
\begin{eqnarray}
{\one} = \begin{pmatrix} 1 & 0 \\ 0 & 1\end{pmatrix},\quad
X = \begin{pmatrix} 0 & 1 \\ 1 & 0\end{pmatrix},\quad
Y = \begin{pmatrix} 0 & -i \\ i & 0\end{pmatrix},\quad
Z = \begin{pmatrix} 1 & 0 \\ 0 & -1\end{pmatrix}.
\end{eqnarray}
It is straightforward to verify that these matrices are both Hermitian and unitary.  
The {computational basis vectors} $\ket{0}, \ket{1} \in \CC^2$ are eigenvectors of the Pauli $Z$ operator, satisfying 
$Z\ket{x} = (-1)^{x}\ket{x}$ for $x\in\ZZ_2$.  
A measurement in the $Z$-basis is described by the projectors
\begin{equation}\label{eq:Z projectors}
\Pi_Z^r := \ket{r}\bra{r} = \tfrac{1}{2}(\one + (-1)^r Z).
\end{equation}
The Hadamard and $T$ gates are given by
\[
H=\tfrac{1}{\sqrt 2} \begin{pmatrix} 1 & 1 \\ 1 & -1\end{pmatrix},
\qquad
T=\begin{pmatrix} 1 & 0 \\ 0 & e^{i\pi/4}\end{pmatrix},
\]
and the controlled-$Z$ gate by
\[
E = \ket{0}\!\bra{0}\otimes \one +
    \ket{1}\!\bra{1}\otimes Z.
\]

\begin{defn}\label{def:qcirc double functor}
The double functor 
\[
\begin{tikzcd}
\phi_\cC: &[-3em]	 \sf{DPG}_{\cC} \arrow[r]  & \sf{QBit}
\end{tikzcd}
\]
is determined by sending the labels  $N_Z$, $H$, $T$, $E$, and $M_Z$ to the corresponding instruments: 
\begin{itemize}
\item The \emph{$Z$-state preparation instrument}, for $r\in \ZZ_2$ and $\chi\in \CC$,
\[
{(\Phi_{N_Z})_r}(\chi) = \chi\ket{r}\bra{r}.
\]
\item The \emph{$U$-gate instruments},
\begin{equation}\label{eq:U gate instrument}
\Phi_U(-) = U(-)U^\dagger,
\end{equation}
where $U\in\{H,T,E\}$.
\item The \emph{(destructive) $Z$-measurement instrument}, for {$s\in \ZZ_2$},
\[
\Phi_{M_Z}^s(-) = \on{Tr}(\Pi_Z^s(-)\Pi_Z^s).
\]
\end{itemize}
{The double category determined by the image of $\phi_\cC$ will be denoted by $\sf{QC}$.}
\end{defn}

{
\begin{rmk}\label{rem:universality and reversible}
By considering the horizontal and vertical monoidal categories of $\sf{QC}$, we recover two fundamental results from the theory of quantum circuits. We define the following categories:
\begin{itemize}
\item $\sf{QUni}$ denotes the wide subcategory of $\sf{QChan}$ whose morphisms are precisely the unitary channels.

\item $\sf{Rev}$ denotes the wide subcategory of $\sf{Bool}$ whose morphisms are precisely the bijective Boolean maps.
\end{itemize}
Let $\sf{H}(\sf{QC})$ denote the horizontal monoidal category of $\sf{QC}$. For each $n\geq 0$, let
\[
G_n\subset \sf{H}(\sf{QC})(n,n)
\]
denote the subgroup generated by the unitary channels corresponding to the gates $H$, $T$, and $E$, acting on arbitrary individual qubits and pairs of qubits. Then $G_n$ is dense in the group of unitary channels on $(\CC^2)^{\otimes n}$, which is canonically isomorphic to the projective unitary group:
\[
\sf{QUni}\big((\CC^2)^{\otimes n},(\CC^2)^{\otimes n}\big)
\cong
PU\!\big((\CC^2)^{\otimes n}\big).
\]
This density expresses quantum universality. The Solovay--Kitaev theorem~\cite{nielsen2010quantum} further shows that arbitrary $n$-qubit unitary channels can be efficiently approximated by circuits over the gate set $H,T,E$. In the vertical direction, we recover the well-known fact that quantum computation encompasses classical reversible computation: $\sf{Rev}$ is a subcategory of the vertical monoidal category $\sf{V}(\sf{QC})$.
\end{rmk}
}

\section{{Adaptive quantum computation}}
\label{sec:adaptive quantum computation}	
	
In this section, we present prominent adaptive models of quantum computation from a double categorical perspective, 
and describe the relationships among them in terms of double functors.

\subsection{{Measurement-based quantum computation}}
\label{sec:MBQC}

Measurement-based quantum computation (MBQC) is a computational model introduced by Raussendorf and Briegel~\cite{raussendorf2001one}, in which computation is carried out through a sequence of adaptive quantum measurements. In this section, we describe this model using our double categorical framework.

We follow the computational framework known as the \emph{measurement calculus}, introduced in~\cite{danos2007measurement}. This framework specifies the basic quantum operations of MBQC, which we use as labels for double port graphs.

\begin{defn}\label{def:MBQC labels}
The set $\mM$ of \emph{MBQC labels} consists of the affine Boolean labels
$\bB_\oplus$, together with the \emph{$X$-preparation},
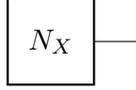
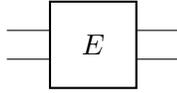
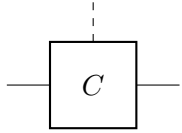
\emph{entangling}, \emph{$C$-correction}, and
\emph{(destructive) $\alpha$-measurement} {labels}, where
$\alpha\in[0,2\pi)$:
\begin{center}
	\begin{tikzpicture}[x=0.5em,y=0.5em]
		\Block{0}{0}{${N_X}$}{N}
		\HOutdraw{N}{1}{1}{3}
	\end{tikzpicture}
	\hspace{2em}
	\begin{tikzpicture}[x=0.5em,y=0.5em]
		\Block{0}{0}{${E}$}{E}
		\HIndraw{E}{2}{1}{3}
		\HIndraw{E}{2}{2}{3}
		\HOutdraw{E}{2}{1}{3}
		\HOutdraw{E}{2}{2}{3}
	\end{tikzpicture}
	\hspace{2em}
	\begin{tikzpicture}[x=0.5em,y=0.5em]
		\Block{0}{0}{${C}$}{C}
		\VIndraw{C}{1}{1}{3}
		\HIndraw{C}{1}{1}{3}
		\HOutdraw{C}{1}{1}{3}
	\end{tikzpicture}
	\hspace{2em}
	\raisebox{-1.5em}{
	\begin{tikzpicture}[x=0.5em,y=0.5em]
		\Block{0}{0}{$M_{\alpha}$}{M}
		\VIndraw{M}{1}{1}{3}
		\HIndraw{M}{1}{1}{3}
		\VOutdraw{M}{1}{1}{3}
	\end{tikzpicture}
	}
\end{center}
The double category of \emph{MBQC-labeled double port graphs} is
$\sf{DPG}_{\mM}$.
\end{defn}

{For the definition of the double functor {into $\sf{QBit}$}, we will need {some preliminary definitions. Consider} the} 
vectors of the form%
\[
\left | +_\alpha^{r} \right \rangle := \frac{1}{\sqrt{2}}\left (\left | 0 \right \rangle + (-1)^r e^{i\alpha}\left | 1 \right \rangle\right ),
\]
where $\alpha \in [0,2\pi)$. 
{We will use the corresponding projectors%
\begin{equation}\label{eq:alpha measurement}
\Pi_{\alpha}^r := \left | +_\alpha^{r} \right \rangle \left \langle +_\alpha^{r} \right | 
= \frac{1}{2}\left ( I + (-1)^{r}(\cos\alpha \, X + \sin \alpha \, Y) \right ).
\end{equation}}%
In particular, the eigenvectors of the Pauli $X$ operator satisfy 
$X\left | \pm \right \rangle = \pm \left | \pm \right \rangle$, 
where $\left | + \right \rangle := \left | +_0^{0} \right \rangle$ 
and $\left | - \right \rangle := \left | +_0^{1} \right \rangle$.

	\begin{defn}\label{def:MBQC double functor} 
The double functor 
\[
\begin{tikzcd}
\phi_\mM: &[-3em] 
\sf{DPG}_{\mM} \arrow[r] &    
\sf{QBit}
\end{tikzcd}
\]
is determined by sending the labels 
 $N_X$, $E$, $C$, and $M_\alpha$ to the corresponding instruments:
\begin{itemize}
    \item The \emph{$X$-preparation instrument}, for $\chi\in \CC$,
    \[
    \Phi_{N}(\chi) = \chi \ket{+}\bra{+}.
    \] 
    \item The \emph{$E$-gate instrument} $\Phi_E$, see Eq.~(\ref{eq:U gate instrument}).   

    \item The \emph{$C$-correction instruments},  
    for $r \in \ZZ_2$,  
    \begin{equation}\label{eq:C correction}
    (\Phi_{C})_r(-) = C^{r} (-) C^{r},
    \end{equation}
    {where $C \in \{X,Z\}$.}

    \item The \emph{(destructive) $\alpha$-measurement instrument}, for $(r,s) \in \ZZ_2^2$, 
    \[
    (\Phi_{M_\alpha})_r^s(-) = 
    \on{Tr}\!\big(\Pi_{(-1)^{r}\alpha}^{s}(-)\Pi_{(-1)^{r}\alpha}^{s}\big).
    \]
\end{itemize}
{The affine Boolean labels are sent to their corresponding Boolean instruments.}
\end{defn}

\begin{eg}\label{ex:teleportation}
The teleportation protocol {\cite{bennett1993teleporting}, see also \cite{danos2007measurement},} can be implemented by a $\mM$-labeled  
{$
\left(
\begin{smallmatrix}
& 0 & \\
1 & & 1 \\
& 0 &
\end{smallmatrix}
\right)
$}-double port graph of the form
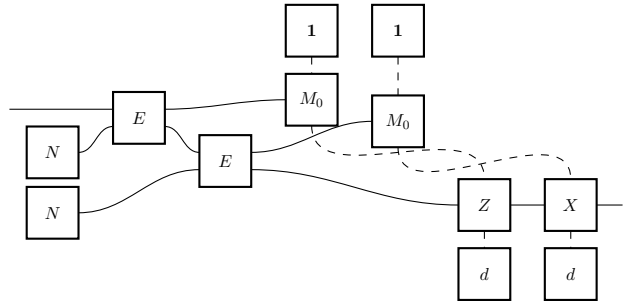
\begin{center}
\begin{tikzpicture}[x=0.5em,y=0.5em]
\def\xshift{10}
\def\vshift{10}
\def\del{2.1}
\begin{scope}[scale=0.8,transform shape]

\Block{\xshift}{0}{$E$}{EM}
\Block{\xshift-\vshift}{5}{$E$}{EL}
\Block{\xshift-2*\vshift}{1}{$N_X$}{NT}
\Block{\xshift-2*\vshift}{-6}{$N_X$}{ND}
\Block{\xshift+\vshift}{5+\del}{$M_0$}{MT}
\Block{\xshift+\vshift}{13+\del}{$\mathbf{0}$}{OL}
\Block{\xshift+2*\vshift}{2.5+\del}{$M_0$}{MD}
\Block{\xshift+2*\vshift}{13+\del}{$\mathbf{0}$}{OR}
\Block{\xshift+3*\vshift}{-3-\del}{$Z$}{Z}
\Block{\xshift+3*\vshift}{-11-\del}{$d$}{DL}
\Block{\xshift+4*\vshift}{-3-\del}{$X$}{X}
\Block{\xshift+4*\vshift}{-11-\del}{$d$}{DR}

\HIndraw{EL}{2}{1}{12}
\HDraw{NT}{1}{1}{EL}{2}{2}{0}{180} 
\HDraw{ND}{1}{1}{EM}{2}{2}{0}{180}
\HDraw{EL}{2}{2}{EM}{2}{1}{0}{180}
\HDraw{EL}{2}{1}{MT}{1}{1}{0}{180}
\HDraw{EM}{2}{1}{MD}{1}{1}{0}{180}
\HDraw{EM}{2}{2}{Z}{1}{1}{0}{180}
\HDraw{Z}{1}{1}{X}{1}{1}{0}{180}  
\HOutdraw{X}{1}{1}{3}   
   
\VDraw{OL}{1}{1}{MT}{1}{1}{-90}{90}   
\VDraw{OR}{1}{1}{MD}{1}{1}{-90}{90}
\VDraw{MT}{1}{1}{Z}{1}{1}{-90}{90}
\VDraw{MD}{1}{1}{X}{1}{1}{-90}{90} 
\VDraw{Z}{1}{1}{DL}{1}{1}{-90}{90}      
\VDraw{X}{1}{1}{DR}{1}{1}{-90}{90}
   
   \end{scope}
\end{tikzpicture}
\end{center} 
Under the $\phi_\mM$ double functor this labeled port graph is sent to an instrument which implements the identity channel on a single qubit.
\end{eg}

{As a consequence of the Solovay--Kitaev theorem (Remark~\ref{rem:universality and reversible}), quantum computational power can be achieved using a finite universal gate set. The counterpart of this result in MBQC is that one can likewise restrict to a finite set of measurement operations.}

\begin{defn}\label{def:pi4 measurement}
Let $\mM[\pi/4]$ denote the subset of $\mM$ where $\alpha$ is restricted to $0, \pm\pi/4$:
\[
\mM[\pi/4]=\{{N_X},E,C,M_\alpha \colon C=X,Z \text{ and } \alpha=0,\pm\pi/4 \}.
\]	   
{The double category determined by the image of the restriction of $\phi_\mM$ to $\sf{DPG}_{\mM[\pi/4]}$ will be denoted by $\sf{MBQC}$.}
\end{defn}

{The QC model can be compared to the MBQC model.} 
We begin by constructing morphisms of port labels:
 	\[	\begin{tikzcd}
\varphi_{\cC,\mM[\pi/4]}:&[-3em]\cC \arrow[r] & \on{Sq}(\sf{DPG}_{\mM[\pi/4]})  .
	\end{tikzcd}
	\] 
	{For the definition of} $\varphi$ {we need} the following fundamental $\mM[\pi/4]$-labeled double port graph; see \cite{danos2007measurement}:

\begin{defn}\label{def:Jalpha box}
The \emph{$J(\alpha)$-gadget} is defined to be the following $\mM$-labeled double port graph: 
\begin{center}
\begin{tikzpicture}[x=0.5em,y=0.5em]
\def\xshift{15}
\def\vshift{10}
\def\del{2.1}
\begin{scope}[scale=1,transform shape]

\Block{\xshift-\vshift}{0}{$J(\alpha)$}{J}
\Block{\xshift+\vshift}{-5}{$N_X$}{N}
\Block{\xshift+2*\vshift}{0}{$E$}{E}
\Block{\xshift+3*\vshift}{+3}{$M_\alpha$}{M}
\Block{\xshift+3*\vshift}{+11}{$\mathbf{1}$}{O}
\Block{\xshift+4*\vshift}{-5}{$X$}{X}
 
\HDraw{N}{1}{1}{E}{2}{2}{0}{180} 
\HDraw{E}{2}{2}{X}{1}{1}{0}{180}
\HDraw{E}{2}{1}{M}{1}{1}{0}{180}   
   
\VDraw{O}{1}{1}{M}{1}{1}{-90}{90}  
\VDraw{O}{1}{1}{M}{1}{1}{-90}{90}  
\VDraw{M}{1}{1}{X}{1}{1}{-90}{90}   
   
\HIndraw{E}{2}{1}{10}   
\HOutdraw{X}{1}{1}{3}
\HIndraw{J}{1}{1}{3}   
\HOutdraw{J}{1}{1}{3}

  \path (16.5,0) node {\huge $=$};  
  \end{scope} 
\end{tikzpicture}
\end{center} 
\end{defn}

{For each target label or gadget $L$, let $\Gamma_L$ denote the corresponding
$\mM[\pi/4]$-labeled double port graph. Then $\varphi$ is defined on the gate
labels by
\[
\varphi(U)=
\begin{cases}
\Gamma_{J(0)} & U=H,\\
\Gamma_{J(0)}\circ\Gamma_{J(\pi/4)} & U=T,\\
\Gamma_E & U=E.
\end{cases}
\]
We write
\[
\Gamma_H:=\Gamma_{J(0)},
\qquad
\Gamma_T:=\Gamma_{J(0)}\circ\Gamma_{J(\pi/4)}
\]
for the $\mM[\pi/4]$-labeled double port graphs implementing the $H$ and $T$
gates, respectively.} 
The remaining labels can be defined as follows: 
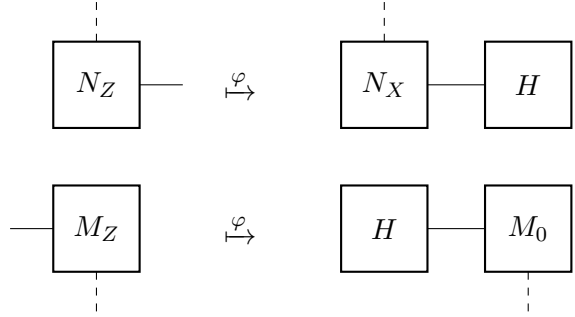
\begin{center}
\begin{tikzpicture}[x=0.5em,y=0.5em]
\def\xshift{30}
\def\hshift{10}
\def\vshift{-10}
\def\del{2.1}

\Block{\xshift-\hshift}{0}{$N_Z$}{NZ}
\Block{\xshift+\hshift}{0}{$N_X$}{NX}
\Block{\xshift+2*\hshift}{0}{$H$}{H} 
 
\VIndraw{NZ}{1}{1}{3}
\HOutdraw{NZ}{1}{1}{3}
\VIndraw{NX}{1}{1}{3} 

\HDraw{NX}{1}{1}{H}{1}{1}{0}{180}   
   
  \path (\xshift,0) node {$\xmapsto{\varphi}$};

\Block{\xshift-\hshift}{0+\vshift}{$M_Z$}{MZq}
\Block{\xshift+\hshift}{0+\vshift}{$H$}{Hq}
\Block{\xshift+2*\hshift}{0+\vshift}{$M_0$}{MZeq} 
 
\VOutdraw{MZq}{1}{1}{3}
\HIndraw{MZq}{1}{1}{3}
\VOutdraw{MZeq}{1}{1}{3} 

\HDraw{Hq}{1}{1}{MZeq}{1}{1}{0}{180}   
   
  \path (\xshift,0+\vshift) node {$\xmapsto{\varphi}$};   
   
\end{tikzpicture}
\end{center}

Now, the pasting construction of Section \ref{sec:pasting} and the functoriality of label morphisms as discussed in Section \ref{sec:functoriality port labels} gives a double functor
 	\begin{equation}\label{eq:C-M}
 		\begin{tikzcd}
\kappa_{\cC,\mM[\pi.4]}:&[-3em] \sf{DPG}_\cC \arrow[r,"\overline \varphi"] & \sf{DPG}_{\on{Sq}(\sf{DPG}_{\mM[\pi/4]})} \arrow[r,"\on{paste}"]& \sf{DPG}_{\mM[\pi/4]} .
	\end{tikzcd}
 	\end{equation}
	 This formalizes what is meant by a quantum circuit implemented by an MBQC.  
Furthermore, the operational meaning is preserved; that is, we have a commutative diagram of double functors
\begin{equation}\label{eq:C-M-QBit}
\begin{tikzcd}[column sep=huge,row sep=large]
\sf{DPG}_\cC \arrow[r,"\kappa"] \arrow[dr,"\phi"'] & \sf{DPG}_{\mM[\pi/4]} \arrow[d,"\phi"]\\
& \sf{QBit}
\end{tikzcd}
\end{equation}

{ 
Typically, an MBQC consists of a state preparation followed by a sequence of adaptive measurements that together implement the computation. 
This structure is formalized in Proposition~\ref{pro:standard form} via the notion of a \emph{standard form}, which expresses the rewrite rules of~\cite{danos2007measurement} in terms of double categorical composition.
}

\subsection{{Quantum computation with magic states}}
\label{sec:QCM}

In this section, we introduce a double category for quantum computation with magic states (QCM), one of the well-known models of universal quantum computation introduced in \cite{bravyi2005universal}.

\begin{defn}\label{def:QCM labels}
The set $\qQ$ of \emph{QCM labels} consists of the affine Boolean labels {$\bB_{\oplus}$} together with the \emph{$T$-preparation}, \emph{entangling}, \emph{$S$-correction}, \emph{$H$-correction}, \emph{(nondestructive) $Z$-measurement}, and \emph{trace} {labels}:
\begin{center}
	\begin{tikzpicture}[x=0.5em,y=0.5em]
		\Block{0}{0}{${N_T}$}{N}
		\HOutdraw{N}{1}{1}{3}
	\end{tikzpicture}
	\hspace{1em}
	\begin{tikzpicture}[x=0.5em,y=0.5em]
		\Block{0}{0}{${E}$}{E}
		\HIndraw{E}{2}{1}{3}
		\HIndraw{E}{2}{2}{3}
		\HOutdraw{E}{2}{1}{3}
		\HOutdraw{E}{2}{2}{3}
	\end{tikzpicture}
	\hspace{1em}
	\begin{tikzpicture}[x=0.5em,y=0.5em]
		\Block{0}{0}{$S$}{S}
		\VIndraw{S}{1}{1}{3}
		\HIndraw{S}{1}{1}{3}
		\HOutdraw{S}{1}{1}{3}
	\end{tikzpicture}
	\hspace{1em}
	\begin{tikzpicture}[x=0.5em,y=0.5em]
		\Block{0}{0}{$H$}{H}
		\VIndraw{H}{1}{1}{3}
		\HIndraw{H}{1}{1}{3}
		\HOutdraw{H}{1}{1}{3}
	\end{tikzpicture}
	\hspace{1em}
	\raisebox{-1.5em}{
	\begin{tikzpicture}[x=0.5em,y=0.5em]
		\Block{0}{0}{$\widetilde M_Z$}{M}
		\HIndraw{M}{1}{1}{3}
		\VOutdraw{M}{1}{1}{3}
		\HOutdraw{M}{1}{1}{3}
	\end{tikzpicture}
	}
	\hspace{1em}
	\begin{tikzpicture}
		\draw (-0.5,-0.5) rectangle (0.5,0.5);
		\path (0,0) node {$\on{Tr}$};
		\draw (-0.5,0) -- (-1,0);
	\end{tikzpicture}
\end{center}
The double category of \emph{QCM-labeled double port graphs} is denoted by $\sf{DPG}_{\qQ}$.
\end{defn}

\begin{defn}
We define a double functor
\[
		\begin{tikzcd}
			 \phi_\qQ: &[-3em] \sf{DPG}_\qQ \arrow[r] & \sf{QBit}
		\end{tikzcd}
		\]
by sending the labels $N_T$, $E$, $S/H$, $\widetilde{M}_Z$, $\on{Tr}$ to the respective instruments:{
\begin{itemize}
\item The \emph{$T$-preparation instrument}, for $\chi\in \CC$,
\[
\Phi_{N_T}(\chi) = \chi T\ket{+}\bra{+}T^\dagger.
\]
\item The \emph{$E$-gate instrument:} $\Phi_E$, see Eq.~(\ref{eq:U gate instrument}).

\item The \emph{$C$-correction instruments} with $C=S$ or $H$, see Eq.~(\ref{eq:C correction}). 

\item The \emph{(non-destructive) measurement instrument}, for $r\in \ZZ_2$,
\[
(\Phi_{\widetilde{M}_Z})^r(-) = \Pi_Z^r(-)\Pi_Z^r
\]
where $\Pi_Z^r$ is as given in Eq.~(\ref{eq:Z projectors}).

\item The \emph{trace instrument} given by the trace map $\on{Tr}:L(\CC^2)\to L(\CC)$.
\end{itemize}
The affine Boolean labels are sent to their corresponding Boolean instruments.
}
\end{defn} 

{
Let $\Gamma_{\widetilde M_Z}$ and $\Gamma_{\on{Tr}}$ denote the
$\qQ$-labeled double port graphs consisting of single vertices labeled by
$\widetilde M_Z$ and $\on{Tr}$, respectively. We define the
\emph{destructive $Z$-measurement} by
\[
\Gamma_{M_Z}
:=
\Gamma_{\on{Tr}}\circ\Gamma_{\widetilde M_Z}.
\]
Thus, the quantum output of the nondestructive measurement is discarded by
the trace operation, while its classical measurement output remains as the
vertical output. 
} 

\begin{defn}
\label{def:T gadget}
The \emph{$T$-gadget} (see \cite{peres2203quantum}) is defined to be the following $\qQ$-labeled double port graph:
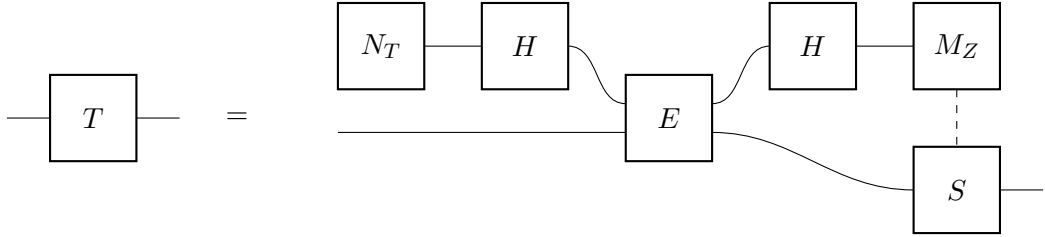
\begin{center}
\begin{tikzpicture}[x=0.5em,y=0.5em]
\def\xshift{30}
\def\hshift{10}
\def\vshift{30}
\def\del{2.1}

\Block{\xshift-\hshift}{0}{$T$}{T}
\Block{\xshift+\hshift}{5}{$N_T$}{NT} 
\Block{\xshift+2*\hshift}{5}{$H$}{H} 
\Block{\xshift+3*\hshift}{0}{$E$}{E}
\Block{\xshift+4*\hshift}{5}{$H$}{HH} 
\Block{\xshift+5*\hshift}{5}{$M_Z$}{MZ} 
\Block{\xshift+5*\hshift}{-5}{$S$}{S} 

\HIndraw{T}{1}{1}{3}
\HOutdraw{T}{1}{1}{3}
 
\HIndraw{E}{2}{2}{20}  
\HOutdraw{S}{1}{1}{3} 

\HDraw{NT}{1}{1}{H}{1}{1}{0}{180}   
\HDraw{H}{1}{1}{E}{2}{1}{0}{180} 
\HDraw{E}{2}{1}{HH}{1}{1}{0}{180}
\HDraw{HH}{1}{1}{MZ}{1}{1}{0}{180}
\HDraw{E}{2}{2}{S}{1}{1}{0}{180}

\VDraw{MZ}{1}{1}{S}{1}{1}{-90}{90}
   
  \path (\xshift,0) node {$=$};

\end{tikzpicture}
\end{center}  
{The double category determined by the image of $\phi_\qQ$ will be denoted by $\sf{QCM}$.}
\end{defn}
	
We construct a morphism of labels
\[
		\begin{tikzcd}
			\varphi_{\cC,\qQ}: &[-3em] \cC \arrow[r] & \on{Sq}(\sf{DPG}_{\qQ})
		\end{tikzcd}
		\]			
by sending 
$H$-label to the vertical composition {$\Gamma_H\bullet \Gamma_{\bf{1}}$}
and sending
the $T$-label to {the $T$-gadget. The remaining labels are mapped to those with the same names.}
Now, using the {pasting operation (Section \ref{sec:pasting})} we can construct the double functor	
	\begin{equation}\label{eq:C-Q}
		\begin{tikzcd}
\kappa_{\cC,\qQ}: \sf{DPG}_\cC \arrow[r,"\overline \varphi"] & \sf{DPG}_{\on{Sq}(\sf{DPG}_{\qQ})} \arrow[r,"\on{paste}"]& \sf{DPG}_{\qQ}. 
	\end{tikzcd}
	\end{equation}
This formalizes conversion of a quantum circuit (QC) into the QCM model. 	
	   
Similar to the MBQC case, we have a commutative diagram of double functors
\begin{equation}\label{eq:C-Q-QBit}
\begin{tikzcd}[column sep=huge,row sep=large]
\sf{DPG}_\cC \arrow[r,"\kappa"] \arrow[dr,"\phi"'] & \sf{DPG}_{\qQ} \arrow[d,"\phi"]\\
& \sf{QBit}
\end{tikzcd}
\end{equation}
Diagrams (\ref{eq:C-M-QBit}) and (\ref{eq:C-Q-QBit}) combine to yield Diagram (\ref{dia:Q-C-M}). 
The corresponding double subcategories assembles into
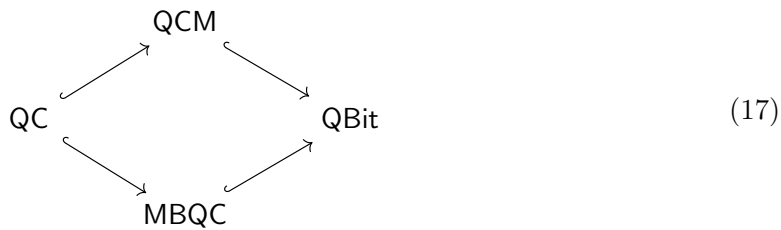
\begin{equation}\label{dia: QC QCM MBQC}
\begin{tikzcd}
  & \sf{QCM} \arrow[dr,hook] &\\
\sf{QC} \arrow[ru,hook]\arrow[dr,hook] && \sf{QBit} \\
&\sf{MBQC} \arrow[ru,hook]&
\end{tikzcd}
\end{equation}
Further conversions are possible between the models QC, QCM, and MBQC; their explicit formulation is left to the reader.
 
\subsection{Measurement-based Pauli computation}
\label{sec:MBPC}

Now, we specialize to the measurement-based computational model introduced in \cite{danos2007pauli} and further developed in \cite{okay2024classical}. We begin with the latter version, which involves fewer basic operations. {We call this model measurement-based Pauli computation (MBPC).}
 
\begin{defn}\label{def:QCM labels}
The set $\pP$ of \emph{MBPC labels} consists of the affine Boolean labels $\bB_\oplus$ together with the \emph{$X/T$-preparation label}, the \emph{entangling label}, and the \emph{(destructive) $X/Y$-measurement label}, respectively:
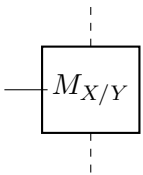
 	\begin{center}
 		\begin{tikzpicture}[x=0.5em,y=0.5em]
 			\Block{0}{0}{${N_{X/T}}$}{N}
 			\HOutdraw{N}{1}{1}{3}
 			\VIndraw{N}{1}{1}{3}
 		\end{tikzpicture} 
\hspace{3em} 
 		\begin{tikzpicture}[x=0.5em,y=0.5em]
 			\Block{0}{0}{${E}$}{E}
 			\HIndraw{E}{2}{1}{3}
 			\HIndraw{E}{2}{2}{3}
 			\HOutdraw{E}{2}{1}{3}
 			\HOutdraw{E}{2}{2}{3}
 		\end{tikzpicture}
\hspace{3em}
\raisebox{-1.5em}{
 		\begin{tikzpicture}[x=0.5em,y=0.5em]
 			\Block{0}{0}{$M_{X/Y}$}{C}
 			\VIndraw{C}{1}{1}{3}
 			\HIndraw{C}{1}{1}{3}
 			\VOutdraw{C}{1}{1}{3}
 		\end{tikzpicture}
 		}
 	\end{center}    
The double category of \emph{MBPC-labeled double port graphs} is denoted by $\sf{DPG}_{\pP}$.
\end{defn}

We begin by describing the $H$-gate label, a special case of the $J(\alpha)$-gadget (Definition \ref{def:Jalpha box}) with $\alpha=0$: 
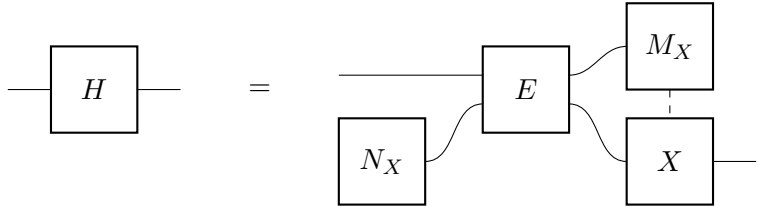
\begin{center}
\begin{tikzpicture}[x=0.5em,y=0.5em]
\def\xshift{15}
\def\vshift{10}
\def\del{2.1}

\Block{\xshift-\vshift}{0}{$H$}{J}
\Block{\xshift+\vshift}{-5}{$N_X$}{N}
\Block{\xshift+2*\vshift}{0}{$E$}{E}
\Block{\xshift+3*\vshift}{+3}{$M_{X}$}{M}
\Block{\xshift+3*\vshift}{-5}{$X$}{X}
 
\HDraw{N}{1}{1}{E}{2}{2}{0}{180} 
\HDraw{E}{2}{2}{X}{1}{1}{0}{180}
\HDraw{E}{2}{1}{M}{1}{1}{0}{180}   
    
\VDraw{M}{1}{1}{X}{1}{1}{-90}{90}   
   
\HIndraw{E}{2}{1}{10}   
\HOutdraw{X}{1}{1}{3}
\HIndraw{J}{1}{1}{3}   
\HOutdraw{J}{1}{1}{3}

  \path (16.5,0) node {$=$};   
\end{tikzpicture}
\end{center} 
Here, and henceforth we write $M_X$ or $M_Y$ for the vertical composition of $M_{X/Y}$ with the $\bf{0}$- or $\bf{1}$-state label. Similarly, we write $N_X$ and $N_T$ for the preparation labels obtained from $N_{X/T}$.

\begin{defn}\label{def:phi P}
We define a double functor  
\[
		\begin{tikzcd}
			 \phi_\pP: &[-3em] \sf{DPG}_\pP \arrow[r] &  \sf{QBit}
		\end{tikzcd}
		\]
		by sending the labels $N_{X/T}$, $E$, and $M_{X/Y}$ to the respective instruments:
\begin{itemize}
\item The \emph{$X/T$-preparation instrument}, for $\chi\in \CC$,
\[
(\Phi_{N_{X/T}})_s(\chi) = \chi T^s\ket{+}\bra{+}(T^s)^\dagger.
\]
\item The \emph{$E$-gate instrument:} $\Phi_E$, see Eq.~(\ref{eq:U gate instrument}).
 
\item The \emph{(destructive) $X/Y$-measurement instrument}, for $s,r\in \ZZ_2$,
\[
(\Phi_{{M}_{X/Y}})_s^r(-) = \on{Tr}(S^s\Pi_{X}^r(S^s)^\dagger (-)S^s\Pi_{X}^r(S^s)^\dagger) 
\]
where $\Pi_X^r=(\one+(-1)^rX)/2$, which also coincides with $\Pi_\alpha$ in Eq.~(\ref{eq:alpha measurement}) when $\alpha=0$.
\end{itemize}
The affine Boolean labels are sent to their corresponding Boolean instruments.
{The double category determined by the image of $\phi_\pP$ will be denoted by $\sf{MBPC}$.}
\end{defn}

A key construction in this model is the realization of the OR-gate label using affine Boolean and quantum labels. This construction is crucial for understanding the vertical direction of the double category {$\sf{MBPC}$}. {In addition, this non-affine Boolean label is used in the model conversion from the circuit model.}

\begin{defn}
\label{def:Anders Browne gadget}
The \emph{{OR}-gadget} {(see \cite{anders2009computational,
raussendorf2013contextuality})} is defined by 
\begin{center}
\begin{tikzpicture}[x=0.5em,y=0.5em]
\def\xshift{-20}
\def\vshift{10}
\def\hshift{10}
\def\del{2.1}
\begin{scope}[scale=0.8,transform shape]

\Block{\xshift-\vshift}{0}{$\on{OR}$}{OR}

\Block{\xshift+\vshift}{0}{$\on{GHZ}$}{G}
\Block{\hshift+\xshift+2*\vshift}{10}{$M_{X/Y}$}{ML}
\Block{\hshift+\xshift+2*\vshift}{2}{$c$}{C}

\Block{\hshift+\xshift+2.3*\vshift}{-9}{$\on{XOR}$}{XM}
\Block{\hshift+\xshift+3*\vshift}{-17}{$\on{XOR}$}{XD}
\Block{\hshift+\xshift+3*\vshift}{-25}{$\on{NOT}$}{NOT}

\Block{\hshift+\xshift+3.9*\vshift}{4}{$\on{XOR}$}{XT}
\Block{\hshift+\xshift+3.5*\vshift}{-5}{$M_{X/Y}$}{MR}

\Block{\hshift+\xshift+2.8*\vshift}{18}{$c$}{CU}
\Block{\hshift+\xshift+2.6*\vshift}{26}{$M_{X/Y}$}{MU}

\VIndraw{MU}{1}{1}{3}
\VDraw{MU}{1}{1}{CU}{1}{1}{-90}{90}
\VDraw{CU}{2}{2}{XT}{2}{2}{-90}{90}
\VDraw{CU}{2}{1}{XM}{2}{2}{-90}{90}
\VIndraw{ML}{1}{1}{19}
\VDraw{ML}{1}{1}{C}{1}{1}{-90}{90}
\VDraw{C}{2}{2}{XT}{2}{1}{-90}{90}
\VDraw{C}{2}{1}{XM}{2}{1}{-90}{90}
\VDraw{XT}{1}{1}{MR}{1}{1}{-90}{90}
\VDraw{MR}{1}{1}{XD}{2}{2}{-90}{90}
\VDraw{XM}{1}{1}{XD}{2}{1}{-90}{90}
\VDraw{XD}{1}{1}{NOT}{1}{1}{-90}{90}
\VOutdraw{NOT}{1}{1}{3}

\HDraw{G}{3}{1}{MU}{1}{1}{0}{180}
\HDraw{G}{3}{2}{ML}{1}{1}{0}{180}
\HDraw{G}{3}{3}{MR}{1}{1}{0}{180}

\VIndraw{OR}{1}{1}{3}
\VOutdraw{OR}{2}{1}{3}
\VOutdraw{OR}{2}{2}{3}

\path (-20,0) node {$=$};
\end{scope}
\end{tikzpicture}
\end{center}
where the $n$-qubit Greenberger--Horne--Zeilinger (GHZ) state~\cite{danos2007measurement} is given by
\begin{center}
\begin{tikzpicture}[x=0.8em,y=0.5em]
\def\xshift{-20}
\def\vshift{10}
\def\del{2.1}
\def\diagscale{0.8}     
\def\centerx{25}        

\begin{scope}[scale=0.6,transform shape]

\Block{\xshift+\vshift-\centerx}{19}{$N_X$}{NU}
\Block{\xshift+\vshift-\centerx}{12}{$N_X$}{NM}
\Block{\xshift+\vshift-\centerx}{5}{$N_X$}{ND}

\Block{\xshift+\vshift-\centerx}{-5}{$N_X$}{NNU}
\Block{\xshift+\vshift-\centerx}{-12}{$N_X$}{NND}

\Block{\xshift+2*\vshift-\centerx}{15.5}{$E$}{E}
\Block{\xshift+3*\vshift-\centerx}{12}{$H$}{H}

\Block{\xshift+4*\vshift-\centerx}{8.5}{$E$}{EE}
\Block{\xshift+5*\vshift-\centerx}{5}{$H$}{HH}

\Block{\xshift+6*\vshift-\centerx}{-5}{$H$}{HL}
\Block{\xshift+7*\vshift-\centerx}{-8.5}{$E$}{EL}
\Block{\xshift+8*\vshift-\centerx}{-12}{$H$}{HHL}

\HDraw{NU}{1}{1}{E}{2}{1}{0}{180}
\HOutdraw{E}{2}{1}{30}
\HDraw{NM}{1}{1}{E}{2}{2}{0}{180}
\HDraw{E}{2}{2}{H}{1}{1}{0}{180}
\HDraw{H}{1}{1}{EE}{2}{1}{0}{180}
\HOutdraw{EE}{2}{1}{20}
\HDraw{ND}{1}{1}{EE}{2}{2}{0}{180}
\HDraw{EE}{2}{2}{HH}{1}{1}{0}{180}
\HOutdraw{HH}{1}{1}{15}

\HDraw{NNU}{1}{1}{HL}{1}{1}{0}{180}
\HDraw{HL}{1}{1}{EL}{2}{1}{0}{180}
\HOutdraw{E}{2}{1}{10}
\HDraw{NND}{1}{1}{EL}{2}{2}{0}{180}
\HDraw{EL}{2}{2}{HHL}{1}{1}{0}{180}
\HOutdraw{HHL}{1}{1}{3}

\path (35-\centerx,.5) node {$\ddots$};

\end{scope}
\end{tikzpicture}
\end{center}
\end{defn}

We also introduce an extended version of the label set $\pP$ by including corrections. This model was introduced in \cite{danos2007pauli}.

\begin{defn}\label{def:QCM labels}
The set $\widetilde\pP$ of \emph{MBPC labels with corrections}  
consists of the labels in $\pP$ together with
the \emph{$X$-correction} and \emph{$S$-correction labels}
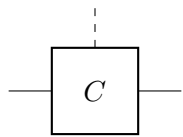
 	\begin{center}
 		\begin{tikzpicture}[x=0.5em,y=0.5em]
 			\Block{0}{0}{$X$}{C} 
 			\VIndraw{C}{1}{1}{3}
 			\HIndraw{C}{1}{1}{3}
 			\HOutdraw{C}{1}{1}{3}
 		\end{tikzpicture}
\hspace{3em}
 		\begin{tikzpicture}[x=0.5em,y=0.5em]
 			\Block{0}{0}{$S$}{C} 
 			\VIndraw{C}{1}{1}{3}
 			\HIndraw{C}{1}{1}{3}
 			\HOutdraw{C}{1}{1}{3}
 		\end{tikzpicture}
 	\end{center}  
 
\end{defn} 

In the extended version we have the following \emph{correction vs measurement rewrite rules}:
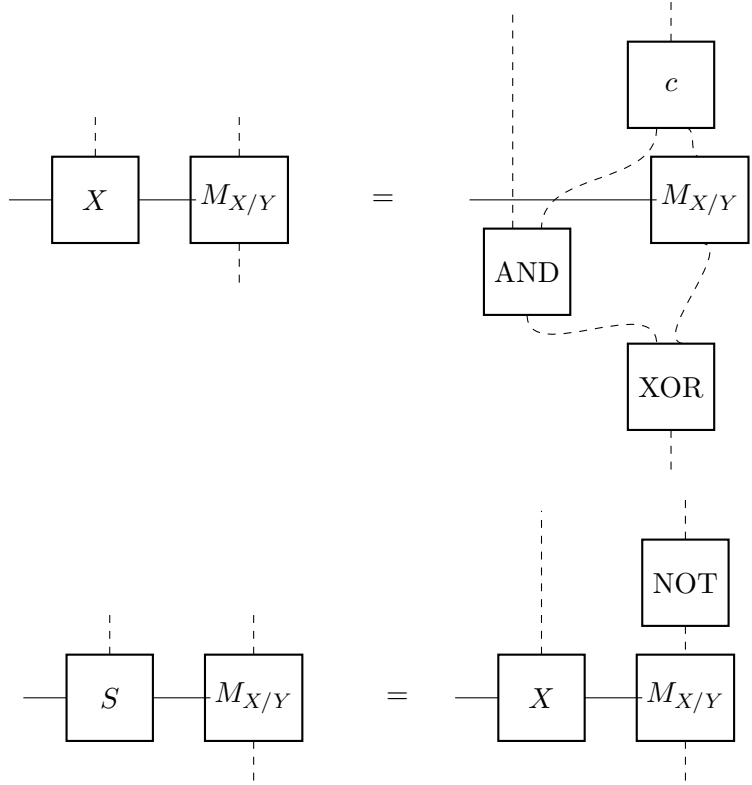
\begin{center}
\begin{tikzpicture}[x=0.5em,y=0.5em]
\def\xshift{30}
\def\hshift{10}
\def\vshift{30}
\def\del{2.1}

\Block{\xshift-\hshift}{0}{$M_{X/Y}$}{MXY}
\Block{\xshift-2*\hshift}{0}{$X$}{X}

\Block{\xshift+\hshift}{-5}{$\on{AND}$}{A}
\Block{\xshift+2.2*\hshift}{0}{$M_{X/Y}$}{MXYR} 
\Block{\xshift+2*\hshift}{-13}{$\on{XOR}$}{XOR} 
\Block{\xshift+2*\hshift}{8}{$c$}{c}

\HIndraw{X}{1}{1}{3} 
\VIndraw{X}{1}{1}{3}
\HDraw{X}{1}{1}{MXY}{1}{1}{0}{180}
\VIndraw{MXY}{1}{1}{3}
\VOutdraw{MXY}{1}{1}{3}
\VDraw{c}{2}{2}{MXYR}{1}{1}{0}{180}
\VDraw{MXYR}{1}{1}{XOR}{2}{2}{0}{180}

\VIndraw{A}{2}{1}{15}
\VDraw{c}{2}{1}{A}{2}{2}{-90}{90}

\VDraw{A}{1}{1}{XOR}{2}{1}{-90}{90}

\VIndraw{c}{1}{1}{3}
\HIndraw{MXYR}{1}{1}{13}

\VOutdraw{XOR}{1}{1}{3}

  \path (\xshift,0) node {$=$};

\end{tikzpicture}
\end{center} 
\begin{center}
\begin{tikzpicture}[x=0.5em,y=0.5em]
\def\xshift{30}
\def\hshift{10}
\def\vshift{30}
\def\del{2.1}

\Block{\xshift-\hshift}{0}{$M_{X/Y}$}{MXY}
\Block{\xshift-2*\hshift}{0}{$S$}{S}

\Block{\xshift+\hshift}{0}{$X$}{X}
\Block{\xshift+2*\hshift}{0}{$M_{X/Y}$}{MXYR} 
\Block{\xshift+2*\hshift}{8}{$\on{NOT}$}{N}

\HIndraw{S}{1}{1}{3} 
\VIndraw{S}{1}{1}{3}
\HDraw{S}{1}{1}{MXY}{1}{1}{0}{180}
\VIndraw{MXY}{1}{1}{3}
\VOutdraw{MXY}{1}{1}{3}

\HIndraw{X}{1}{1}{3}
\HDraw{X}{1}{1}{MXYR}{1}{1}{0}{180} 
\VDraw{N}{1}{1}{MXYR}{1}{1}{-90}{90} 

\VIndraw{N}{1}{1}{3}
\VIndraw{X}{1}{1}{10}
\VOutdraw{MXYR}{1}{1}{3}

  \path (\xshift,0) node {$=$};

\end{tikzpicture}
\end{center} 	
The second rewrite rule can be further expanded using the first one. 
Observe that the crucial point is that the AND-{label} appears, which is {not an affine} Boolean operation. {However, as we have seen in Section \ref{sec:boolean circuits} the AND-label can be implemented using the OR-label, which in turn can be implemented in $\sf{DPG}_\pP\subset \sf{DPG}_{\widetilde{\pP}}$ using the {OR}-gadget.}

\begin{defn}\label{def:phi tilde P}
We extend the double functor $\phi_\pP$ of Definition \ref{def:phi P} to form a commutative diagram of double functors
\[
		\begin{tikzcd}[column sep=huge, row sep=large]
\sf{DPG}_{\widetilde \pP} \arrow[r,"\phi_{\widetilde\pP}"] & \sf{QBit} \\
\sf{DPG}_\pP \arrow[u,hook] \arrow[ru,"\phi_\pP"'] &  
		\end{tikzcd}
		\] 
by extending the corresponding label morphism to 
 $\widetilde{\pP}$ by sending the remaining $C$-correction labels to corresponding $C$-gate instrument in {Eq.~(\ref{eq:U gate instrument})}. 
{The double category determined by the image of $\phi_{\widetilde\pP}$ will be denoted by $\sf{MBPC}_c$.}
\end{defn}

{To compare MBPC with corrections to QCM w}e define {a morphism of port labels}
\[
\begin{tikzcd}
\varphi_{\widetilde{\pP},\qQ}:&[-3em]\widetilde\pP \arrow[r] & \on{Sq}(\sf{DPG}_{\qQ}) 
\end{tikzcd}
\]
by sending {$N_T$, $E$, $S$ labels} to the labels with the same name, and by sending $H$-correction label to the $H$-correction described above and
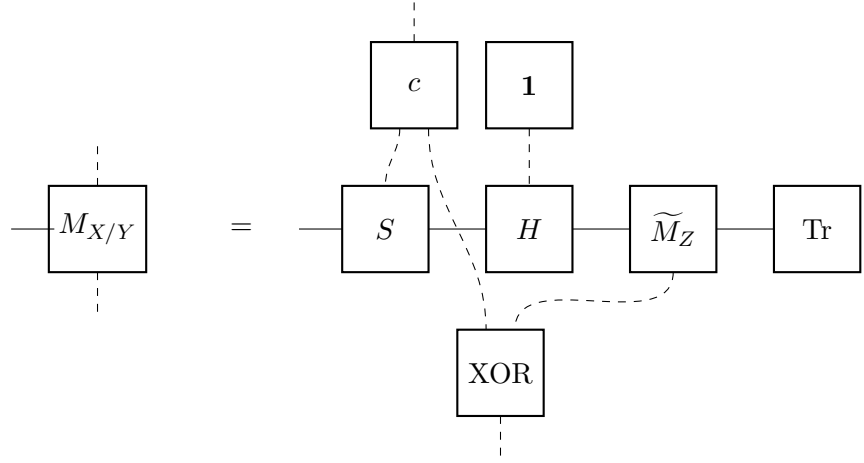
\begin{center}
\begin{tikzpicture}[x=0.5em,y=0.5em]
\def\xshift{-20}
\def\vshift{10}
\def\hshift{10}
\def\del{2.1} 

\Block{\xshift-\vshift}{0}{$M_{X/Y}$}{MXY}
\Block{\xshift+\vshift}{0}{$S$}{S}
\Block{\xshift+1.2*\vshift}{10}{$c$}{C} 
\Block{\xshift+2*\vshift}{0}{$H$}{H}
\Block{\xshift+2*\vshift}{10}{$\mathbf{1}$}{O}
\Block{\xshift+3*\vshift}{0}{$\widetilde{M}_Z$}{MZT}
\Block{\xshift+1.8*\vshift}{-10}{$\on{XOR}$}{XOR}
\Block{\xshift+4*\vshift}{0}{$\on{Tr}$}{T}     

\VDraw{C}{2}{1}{S}{1}{1}{-90}{90} 
\VDraw{C}{2}{2}{XOR}{2}{1}{-90}{90}
\VDraw{O}{1}{1}{H}{1}{1}{-90}{90}
\VDraw{MZT}{1}{1}{XOR}{2}{2}{-90}{90}
\VOutdraw{XOR}{1}{1}{3}
\VIndraw{C}{1}{1}{3}

\VIndraw{MXY}{1}{1}{3}
\VOutdraw{MXY}{1}{1}{3}

\HIndraw{MXY}{1}{1}{3}
\HIndraw{S}{1}{1}{3}
\HDraw{S}{1}{1}{H}{1}{1}{0}{180}
\HDraw{H}{1}{1}{MZT}{1}{1}{0}{180}
\HDraw{MZT}{1}{1}{T}{1}{1}{0}{180}

  \path (-20,0) node {$=$};  
\end{tikzpicture}
\end{center} 	
Composing with the pasting double functor we obtain
 	\begin{equation}\label{eq:P-Q}
 	\begin{tikzcd}
\kappa_{\tilde\pP,\qQ}:&[-3em]\sf{DPG}_{\widetilde\pP} \arrow[r,"\overline\varphi"] & \sf{DPG}_{\on{Sq}(\sf{DPG}_\qQ)} \arrow[r,"\on{paste}"] &  \sf{DPG}_{\qQ}  .
	\end{tikzcd}
 	\end{equation} 		
	 
To compare the image of $\sf{DPG}_\cC$ in $\sf{DPG}_\qQ$, we compose the
$J(\pi/4)$-gadget from Definition~\ref{def:Jalpha box}, with
$\alpha=\pi/4$, with the teleportation diagram from
Example~\ref{ex:teleportation}, and then apply the rewrite rules of
Section~\ref{sec:standard form}. 

\begin{defn}\label{def:teleported J pi 4}
The \emph{teleported $J(\pi/4)$-gadget} {(see \cite{danos2007measurement})} is defined by {the labeled double port graph in Figure \ref{fig:teleported-j-pi-over-4}.}
\end{defn}

Each {occurrence of the} AND operation {in Figure \ref{fig:teleported-j-pi-over-4}}  can again be implemented using the {OR}-gadget. {The $Z$-correction label can be implemented using the $S$-correction label. Then we can implement the $T$-gate label using the labels in $\pP$ by the horizontal composition {$\Gamma_H\circ\Gamma_{J(\pi/4)}$}. Similarly the preparation and measurement labels in $\cC$ can be implemented using the labels in $\widetilde\pP$. This way we obtain a morphism of labels}
\[
\begin{tikzcd}
\varphi_{\cC,\widetilde{\pP}}:&[-3em]\cC \arrow[r] & \on{Sq}(\sf{DPG}_{\widetilde\pP}) 
\end{tikzcd}
\] 
and in turn we obtain a double functor
 	\begin{equation}\label{eq:C-P}
 	\begin{tikzcd}
\kappa_{\cC,\widetilde\pP}:&[-3em]\sf{DPG}_{\cC} \arrow[r,"\overline\varphi"] & \sf{DPG}_{\on{Sq}(\sf{DPG}_{\widetilde \pP})} \arrow[r,"\on{paste}"] &  \sf{DPG}_{\widetilde{\pP}} 
	\end{tikzcd}
 	\end{equation} 	
which allows us to implement a quantum circuit in the MBPC-with-corrections model.
  
Diagrams (\ref{eq:P-Q}) and (\ref{eq:C-P}) assemble into 
Diagram (\ref{dia:Q-C-P}) and the corresponding double subcategories are given by  
\begin{equation}\label{dia:MBPC QC QCM}
\begin{tikzcd}[column sep=huge, row sep=large]
&\sf{QC} \arrow[d,hook] \arrow[rr,hook] && \sf{QBit} \\
\sf{MBPC}\arrow[r,hook]&\sf{MBPC}_c \arrow[rr,hook] && \sf{QCM} \arrow[u,hook] \\ 
\end{tikzcd}
\end{equation}  
{Observe also that $\sf{MBPC}_c$ is a double subcategory of $\sf{MBQC}$, induced by the inclusion of label sets
$\widetilde{\pP}\subset \mM$.}
 
{
\begin{thm}\label{thm:vertical AQC}
Let $\sf{D}$ be one of the double categories $\sf{MBPC}$, $\sf{MBPC}_c$, $\sf{MBQC}$, or $\sf{QCM}$. Then
\[
\sf{Bool}
\subset
\sf{V}(\sf{MBPC}).
\] 
\end{thm}

\begin{proof}
This follows from the inclusions among the double categories established above and from the fact that the OR-gadget in Definition~\ref{def:Anders Browne gadget} implies that any Boolean function can be implemented in the vertical direction of $\sf{MBPC}$.
\end{proof}
}

\bibliography{bib.bib}
\bibliographystyle{ieeetr}
	
\appendix

\section{{Double categories}}
\label{sec:double category}

Double categories admit composition operations in two directions—horizontal
and vertical—which may be either strict or weak
\cite{grandis1999limits}. In this paper, by a double category we mean a
strict double category.

\begin{defn}\label{def:double category}
A \emph{(strict) double category} $\sf{D}$ is an internal category in the category $\sf{Cat}$ of (small) categories.
\end{defn}

Let us unravel this definition. The double category $\sf{D}$ comes with an object category $\sf{D}_0$ and a morphism category $\sf{D}_1$ together with the source and target functors $s,t:\sf{D}_1\to \sf{D}_0$. To restore the hidden symmetry in this structure the convention is to employ the following language:
\begin{itemize}
\item \emph{objects}: $\on{Ob}(\sf{D}_0)$,
\item \emph{{vertical} morphisms}: $\on{Mor}(\sf{D}_0)$,
\item \emph{{horizontal} morphisms}: $\on{Ob}(\sf{D}_1)$,
\item \emph{squares}: $\on{Mor}(\sf{D}_1)$.
\end{itemize} 
A square $\alpha$ 
of type
$
\left(
\begin{smallmatrix}
&f&\\
u&&v\\
&g&
\end{smallmatrix}
\right)
$
is depicted by a diagram of the form
\[
\begin{tikzcd}
c \arrow[r,"f"] \arrow[d,"u"'] & d \arrow[d,"v"] \\
c' \arrow[r,"g"'] & d'
\arrow[phantom, from=1-2, to=2-1, "\alpha" description]
\end{tikzcd}
\]
where $c,c'd,d'$ are the objects, $f,g$ are the horizontal morphisms, and $u,v$ are the vertical morphisms. {Given another} 
$
\left(
\begin{smallmatrix}
&h&\\
v&&w\\
&k&
\end{smallmatrix}
\right)
$-square $\beta$ 
\[
\begin{tikzcd}
d \arrow[r,"h"] \arrow[d,"v"'] & e \arrow[d,"w"] \\
d' \arrow[r,"k"'] & e'
\arrow[phantom, from=1-2, to=2-1, "\beta" description]
\end{tikzcd}
\]
we can form the horizontal composition to obtain 
\[
\begin{tikzcd}
c \arrow[r,"h\circ f"] \arrow[d,"u"'] & e \arrow[d,"w"] \\
c' \arrow[r,"k\circ g"'] & e'
\arrow[phantom, from=1-2, to=2-1, "\beta \circ \alpha" description]
\end{tikzcd}
\] 
{On the other hand}, we can compose vertically with a 
$
\left(
\begin{smallmatrix}
&g&\\
w&&t\\
&h&
\end{smallmatrix}
\right)
$-square {$\theta$}
\[
\begin{tikzcd}
c' \arrow[r,"g"] \arrow[d,"w"'] & d' \arrow[d,"t"] \\
c'' \arrow[r,"h"'] & d''
\arrow[phantom, from=1-2, to=2-1, "\theta" description]
\end{tikzcd}
\]
to obtain
\[
\begin{tikzcd}
c \arrow[r,"f"] \arrow[d,"w\bullet u"'] & d \arrow[d,"t\bullet v"] \\
c'' \arrow[r,"h"'] & d''
\arrow[phantom, from=1-2, to=2-1, "{\theta} \bullet \alpha" description]
\end{tikzcd}
\] 
These compositions satisfy associativity and unitality. The horizontal and vertical units are denoted by
\[
\begin{tikzcd}
c \arrow[r,"\on{Id}_c"] \arrow[d,"u"'] & c \arrow[d,"u"] \\
c' \arrow[r,"\on{Id}_{c'}"'] & c'
\arrow[phantom, from=1-2, to=2-1, "\on{Id}_u^h" description]
\end{tikzcd}
\;\;\;\;
\text{ and }
\;\;\;\;
\begin{tikzcd}
c \arrow[r,"f"] \arrow[d,"1_c"'] & d \arrow[d,"1_d"] \\
c \arrow[r,"f"'] & d
\arrow[phantom, from=1-2, to=2-1, "\on{Id}_f^v" description]
\end{tikzcd}
\]

{The vertical and horizontal compositions of squares satisfy certain compatibility conditions. For our purposes, the main one is the \emph{interchange law}. Given compatible squares}
\[
\begin{tikzcd}
c \arrow[r] \arrow[d] & d \arrow[d] \arrow[r] & e\arrow[d] \\
c' \arrow[d]  \arrow[r] & d' \arrow[r] \arrow[d] & e' \arrow[d]\\
c'' \arrow[r]  & d''  \arrow[r] & e'' 
\arrow[phantom, from=1-2, to=2-1, "\alpha" description]
\arrow[phantom, from=1-3, to=2-2, "\beta" description]
\arrow[phantom, from=2-2, to=3-1, "\theta" description]
\arrow[phantom, from=2-3, to=3-2, "\gamma" description]
\end{tikzcd}
\]
composing first horizontally and then vertically gives the same result as composing first vertically and then horizontally:
\[
(\gamma \circ\theta)\bullet (\beta \circ \alpha) = (\gamma\bullet \beta)
\circ (\theta \bullet \alpha). 
\]
{In addition, the horizontal and vertical units satisfy} $\on{Id}_{1_c}^h = \on{Id}_{\on{Id}_c}^v$.

{The horizontal and vertical data of a double category can be assembled into $2$-categories. The objects of $\sf{D}$, its vertical morphisms, and squares of the form}
\begin{equation}\label{eq:horizontal id}
\left(
\begin{smallmatrix}
&\on{Id}_c&\\
u&&v\\
&\on{Id}_d&
\end{smallmatrix}
\right)
\end{equation} 
form the vertical $2$-category. On the other hand, the objects of $\sf{D}$, its horizontal morphisms, and squares of the form
\begin{equation}\label{eq:vertical id}
\left(
\begin{smallmatrix}
&f&\\
1_c&&1_d\\
&g&
\end{smallmatrix}
\right)
\end{equation}
form the horizontal $2$-category.

{Throughout the paper, we work with $1$-object double categories. In this case, these $2$-categories are determined by monoidal categories. For our purposes, it is more convenient to name these monoidal categories according to the direction in which their squares are composed. Thus, our naming convention is opposite to that of the associated $2$-categories.}

\begin{defn}\label{def:h v monoidal}
For a $1$-object double category $\sf{D}$, we define the following associated strict monoidal categories:
\begin{itemize}
\item The \emph{horizontal monoidal category} $\sf{H}(\sf{D})$ has the vertical morphisms of $\sf{D}$ as objects. A morphism $u \to v$ is a square as in \eqref{eq:horizontal id}.
Composition of morphisms is given by horizontal composition of squares, and the tensor product is given by vertical composition.

\item The \emph{vertical monoidal category} $\sf{V}(\sf{D})$ has the horizontal morphisms of $\sf{D}$ as objects. A morphism $f \to g$ is a square as in \eqref{eq:vertical id}.
Composition of morphisms is given by vertical composition of squares, and the tensor product is given by horizontal composition.
\end{itemize}
\end{defn}
 
{With this convention, the vertical $2$-category is the delooping of the horizontal monoidal category, whereas the horizontal $2$-category is the delooping of the vertical monoidal category.}

\section{{The standard form}}

 \label{sec:standard form}


{Following~\cite{danos2007measurement}, we introduce the following \emph{rewrite rules} to capture the relations in the image of $\phi_{\mM}$, that is, among the instruments corresponding to the MBQC labels. These rules will subsequently be used to define a quotient double category of $\sf{DPG}_{\mM}$.}
\begin{itemize}

\item {\emph{Correction vs. entangling rewrite rules}:}
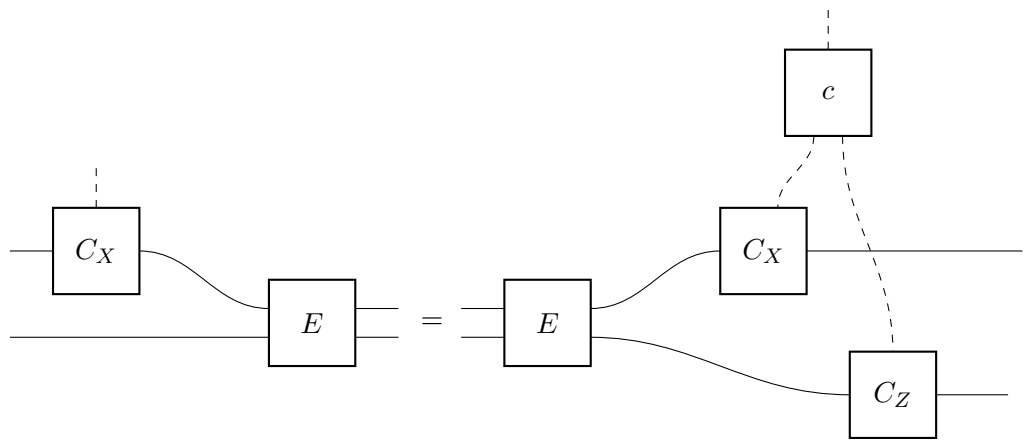
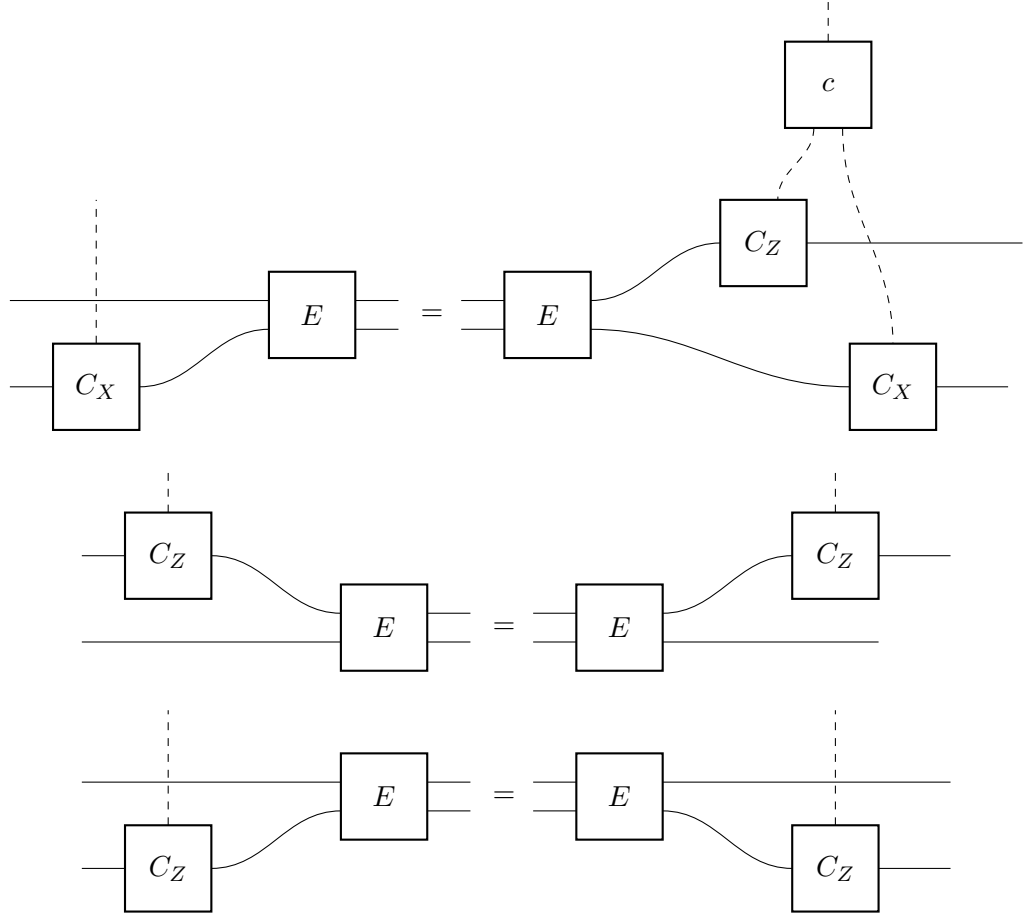
{
\begin{center}\label{diag:X1-E}
	\begin{tikzpicture}[x=0.5em,y=0.5em]

		\Block{15}{5}{$C_{X}$}{X}
		\Block{24}{-5}{$C_{Z}$}{Z}
		\Block{0}{0}{$E$}{E}
		\Block{19.5}{16}{$c$}{C}
		
		\VDraw{C}{2}{2}{Z}{1}{1}{-90}{90}
		\VDraw{C}{2}{1}{X}{2}{2}{-90}{90}
		
		\HOutdraw{X}{1}{1}{15}
		\HOutdraw{Z}{1}{1}{5}
		\HIndraw{E}{2}{1}{3}
		\HIndraw{E}{2}{2}{3}
		\HDraw{E}{2}{1}{X}{1}{1}{0}{180}
		\HDraw{E}{2}{2}{Z}{1}{1}{0}{180}

		\VIndraw{C}{1}{1}{3}

	\begin{scope}[xshift=-6cm]
		\Block{0}{5}{$C_{X}$}{X2}
		\Block{15}{0}{$E$}{E2}
		
		\HDraw{X2}{1}{1}{E2}{2}{1}{0}{180}
		
		\HOutdraw{E2}{2}{1}{3}
		\HOutdraw{E2}{2}{2}{3}
		\HIndraw{E2}{2}{2}{18}
		\HIndraw{X2}{1}{1}{3}
		
		\VIndraw{X2}{1}{1}{3}
		
		\path (23,0) node {~$=$~};
	\end{scope}
\end{tikzpicture}
\end{center}

\begin{center}\label{diag:X2-E}
	\begin{tikzpicture}[x=0.5em,y=0.5em]

		\Block{15}{5}{$C_Z$}{X}
		\Block{24}{-5}{$C_X$}{Z}
		\Block{0}{0}{$E$}{E}
		\Block{19.5}{16}{$c$}{C}
		
		\VDraw{C}{2}{2}{Z}{1}{1}{-90}{90}
		\VDraw{C}{2}{1}{X}{2}{2}{-90}{90}
		
		\HOutdraw{X}{1}{1}{15}
		\HOutdraw{Z}{1}{1}{5}
		\HIndraw{E}{2}{1}{3}
		\HIndraw{E}{2}{2}{3}
		\HDraw{E}{2}{1}{X}{1}{1}{0}{180}
		\HDraw{E}{2}{2}{Z}{1}{1}{0}{180}

		\VIndraw{C}{1}{1}{3}

	\begin{scope}[xshift=-6cm]
		\Block{0}{-5}{$C_X$}{X2}
		\Block{15}{0}{$E$}{E2}
		
		\HDraw{X2}{1}{1}{E2}{2}{2}{0}{180}
		
		\HOutdraw{E2}{2}{1}{3}
		\HOutdraw{E2}{2}{2}{3}
		\HIndraw{E2}{2}{1}{18}
		\HIndraw{X2}{1}{1}{3}
		
		\VIndraw{X2}{1}{1}{10}
		
		\path (23,0) node {~$=$~};
	\end{scope}
\end{tikzpicture}
\end{center}

\begin{center}\label{diag:Z1-E}
	\begin{tikzpicture}[x=0.5em,y=0.5em]

		\Block{15}{5}{$C_Z$}{X}
		\Block{0}{0}{$E$}{E}
		
		
		\HOutdraw{X}{1}{1}{5}
		\HOutdraw{E}{2}{2}{15}
		\HIndraw{E}{2}{1}{3}
		\HIndraw{E}{2}{2}{3}
		\HDraw{E}{2}{1}{X}{1}{1}{0}{180}

		\VIndraw{X}{1}{1}{3}

	\begin{scope}[xshift=-6cm]
		\Block{0}{5}{$C_Z$}{X2}
		\Block{15}{0}{$E$}{E2}
		
		\HDraw{X2}{1}{1}{E2}{2}{1}{0}{180}
		
		\HOutdraw{E2}{2}{1}{3}
		\HOutdraw{E2}{2}{2}{3}
		\HIndraw{E2}{2}{2}{18}
		\HIndraw{X2}{1}{1}{3}
		
		\VIndraw{X2}{1}{1}{3}
		
		\path (23,0) node {~$=$~};
	\end{scope}
\end{tikzpicture}
\end{center}

\begin{center}\label{diag:Z2-E}
	\begin{tikzpicture}[x=0.5em,y=0.5em]

		\Block{15}{-5}{$C_Z$}{Z}
		\Block{0}{0}{$E$}{E}
		
		
		\HOutdraw{E}{2}{1}{20}
		\HOutdraw{Z}{1}{1}{5}
		\HIndraw{E}{2}{1}{3}
		\HIndraw{E}{2}{2}{3}
		\HDraw{E}{2}{2}{Z}{1}{1}{0}{180}

		\VIndraw{Z}{1}{1}{8}

	\begin{scope}[xshift=-6cm]
		\Block{0}{-5}{$C_Z$}{X2}
		\Block{15}{0}{$E$}{E2}
		
		\HDraw{X2}{1}{1}{E2}{2}{2}{0}{180}
		
		\HOutdraw{E2}{2}{1}{3}
		\HOutdraw{E2}{2}{2}{3}
		\HIndraw{E2}{2}{1}{18}
		\HIndraw{X2}{1}{1}{3}
		
		\VIndraw{X2}{1}{1}{8}
		
		\path (23,0) node {~$=$~};
	\end{scope}
\end{tikzpicture}
\end{center}
}
\end{itemize}
It is straightforward to show the correctness of these rewrite rules. We demonstrate this explicitly for the first one. The others follow similarly. The left-hand side corresponds to the instrument $\Phi_{EC} := \Phi_{E}\circ \Phi_{C_{X}}$. For $r\in {\ZZ_2}$ we have that $(\Phi_{EC})_r \in \on{CP}((\CC^{2})^{\otimes 2},(\CC^{2})^{\otimes 2})$ given by%
\begin{align*}
{(\Phi_{EC})_r}(-) &= {E} (X^{r}\otimes I)(-)(X^{r}\otimes I){E}\\
 &= (X^{r}\otimes Z^{r}){E}(-){E}(X^{r}\otimes Z^{r}),
\end{align*}
where the second equality follows from the conjugation action of ${E}$. Consider now the right-hand side of the diagram with corresponding instrument $\Phi_{CE} := ((\Phi_{C_{X}}\otimes \Phi_{C_{z}})\bullet \Phi_{c})\circ \Phi_{E}$. Evaluated on $r\in {\ZZ_2}$ 
we have
\begin{align*}
\Phi_{CE}(r)(-) &= \sum_{s,t}\delta_{r,s}\delta_{r,t} (X^{s}\otimes Z^{t}){E}(-){E}(X^{s}\otimes Z^{t}) \\
&= (X^{r}\otimes Z^{r}){E}(-){E}(X^{r}\otimes Z^{r}),
\end{align*}
where the sum over $(s,t) \in {\ZZ_2^2}$ comes from the vertical composition of the Boolean instrument $\Phi_{c}$ with $\Phi_{C_{X}}\otimes \Phi_{C_{Z}}$.
This is precisely the instrument $\Phi_{EC}$.
 
{
\begin{itemize}

\item {\emph{Correction vs. measurement rewrite rules}:}

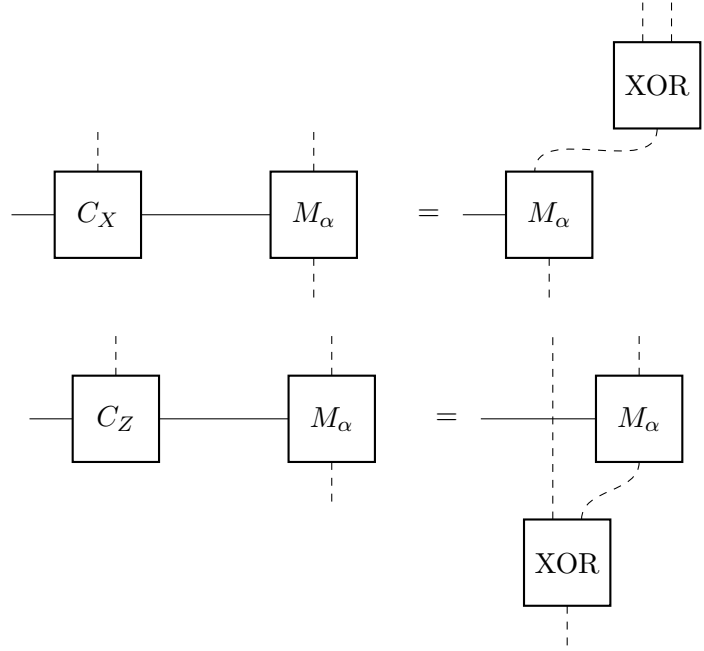
\begin{center}\label{diag:X-M}
	\begin{tikzpicture}[x=0.5em,y=0.5em]
	\Block{0}{0}{$M_\alpha$}{PA}
	\Block{7.5}{9}{$\on{XOR}$}{XOR}
	
	\VDraw{XOR}{1}{1}{PA}{2}{1}{-90}{90}
	
	\VOutdraw{PA}{1}{1}{3}
	\VIndraw{XOR}{2}{1}{3}
	\VIndraw{XOR}{2}{2}{3}
	
	\HIndraw{PA}{1}{1}{3}

	\begin{scope}[xshift=-6cm]
		\Block{0}{0}{$C_X$}{Z2}
		\Block{15}{0}{$M_\alpha$}{PA2}
		
		\HDraw{Z2}{1}{1}{PA2}{1}{1}{0}{180}
		
		\HIndraw{Z2}{1}{1}{3}
		
		\VOutdraw{PA2}{1}{1}{3}
		\VIndraw{PA2}{1}{1}{3}
		\VIndraw{Z2}{1}{1}{3}
		
		\path (23,0) node {$=$};
	\end{scope}
\end{tikzpicture}
\end{center}

\begin{center}\label{diag:Z-M}
	\begin{tikzpicture}[x=0.5em,y=0.5em]
	\Block{5}{0}{$M_\alpha$}{PA}
	\Block{0}{-10}{$\on{XOR}$}{XOR}
	
	\VDraw{PA}{1}{1}{XOR}{2}{2}{-90}{90}
	
	\VOutdraw{XOR}{1}{1}{3}
	\VIndraw{PA}{1}{1}{3}
	\VIndraw{XOR}{2}{1}{13}
	
	\HIndraw{PA}{1}{1}{8}

	\begin{scope}[xshift=-6cm]
		\Block{0}{0}{$C_Z$}{Z2}
		\Block{15}{0}{$M_\alpha$}{PA2}
		
		\HDraw{Z2}{1}{1}{PA2}{1}{1}{0}{180}
		
		\HIndraw{Z2}{1}{1}{3}
		
		\VOutdraw{PA2}{1}{1}{3}
		\VIndraw{PA2}{1}{1}{3}
		\VIndraw{Z2}{1}{1}{3}
		
		\path (23,0) node {$=$};
	\end{scope}
\end{tikzpicture}
\end{center}

\end{itemize}~\\
}%
{%
We now prove the validity of these rewrite rules. Let us begin with the first and introduce the instrument $\Phi_{CM} := \Phi_{M_{\alpha}}\circ\Phi_{C_{X}}$. For $((r,s),t)\in  {\ZZ_2^2\times \ZZ_2}$ we have that%
\begin{eqnarray}
{(\Phi_{CM})_{r,s}^t}(-)%
&=& \text{Tr}\left (\Pi_{(-1)^{s}\alpha}^{t}X^{r}(-)X^{r} \Pi_{(-1)^{s}\alpha}^{t}\right )\notag\\
&=& \text{Tr}\left (X^{r}\Pi_{(-1)^{r+s}\alpha}^{t}(-)\Pi_{(-1)^{r+s}\alpha}^{t} X^{r}\right )\notag\\
&=& \text{Tr}\left (\Pi_{(-1)^{r+s}\alpha}^{t}(-)\Pi_{(-1)^{r+s}\alpha}^{t}\right )\notag\\
&=& {(\Phi_{M_{\alpha}})_{r+s}^t}(-),\notag%
\end{eqnarray}
where in the second equality we used that $X^{r}(\cos\alpha\, X + \sin\alpha\, Y)X^{r} =  \cos\alpha\, X +(-1)^{r} \sin\alpha\, Y = \cos\alpha \, X + \sin (-1)^{r}\alpha \, Y$ and then used the cyclic property of the trace in the third equality. Alternatively we also have the instrument $\Phi_{M\on{X}} := \Phi_{M_{\alpha}}\bullet \Phi_{\on{XOR}}$. Let $((r,s),t)\in {\ZZ_2^2\times \ZZ_2}$ then we have%
\begin{eqnarray}
{(\Phi_{M\on{X}})_{r,s}^t}(-)%
&=& \sum_{u} \delta_{s\oplus r,u}\text{Tr}\left (\Pi_{(-1)^{u}\alpha}^{t}(-)\Pi_{(-1)^{u}\alpha}^{t}\right )\notag\\%
&=& \text{Tr}\left (\Pi_{(-1)^{r+s}\alpha}^{t}(-)\Pi_{(-1)^{r+s}\alpha}^{t}\right )= {(\Phi_{M_{\alpha}})_{r+s}^t}(-).\notag
\end{eqnarray}
}

{
For the second rewrite rule let us define the maps $\Phi_{CM'} := \Phi_{M_{\alpha}}\circ\Phi_{C_{Z}}$ and $\Phi_{\on{X}M} := \Phi_{\on{XOR}}\bullet \Phi_{M_{\alpha}}$. Computing the first map for $((r,s),t)\in {\ZZ_2^2\times \ZZ_2}$ gives
\begin{eqnarray}
{(\Phi_{CM})_{r,s}^t}(-)%
&=& \text{Tr}\left (\Pi_{(-1)^{s}\alpha}^{t}Z^{r}(-)Z^{r} \Pi_{(-1)^{s}\alpha}^{t}\right )\notag\\
&=& \text{Tr}\left (Z^{r}\Pi_{(-1)^{s}\alpha}^{r+t}(-)\Pi_{(-1)^{s}\alpha}^{r+t} Z^{r}\right )\notag\\
&=& \text{Tr}\left (\Pi_{(-1)^{s}\alpha}^{r+t}(-)\Pi_{(-1)^{s}\alpha}^{r+t}\right )\notag\\
&=& {(\Phi_{M_{\alpha}})_{s}^{r+t}}(-),\notag%
\end{eqnarray}
where in the second equality we used that $Z^{r}X = (-1)^{r}XZ^{r}$ and $Z^{r}Y = (-1)^{r}YZ^{r}$ since $Z$ anticommutes with $X$ and $Y$. On the other hand, we also have that
\begin{eqnarray}
{(\Phi_{M\on{X}})_{r,s}^t} (-)%
&=& \sum_{u} \delta_{r\oplus t,u}\text{Tr}\left (\Pi_{(-1)^{s}\alpha}^{u}(-)\Pi_{(-1)^{s}\alpha}^{{u}}\right )\notag\\%
&=& \text{Tr}\left (\Pi_{(-1)^{s}\alpha}^{r+t}(-)\Pi_{(-1)^{s}\alpha}^{r+t}\right )= {(\Phi_{M_{\alpha}})_{s}^{r+t}} (-).\notag
\end{eqnarray}
}
	
{	
For the standard form we introduce the following label sets 
\begin{itemize}
\item \emph{preparation label set} $\mM_P$ consists 
of the $N_X$-label and $E$-label.
\item \emph{measurement label set} $\mM_M$ is the union of $M_\alpha$ where $\alpha=0,\pm \pi/4$, and ${\bB_{\oplus}}$,

\item \emph{correction label set} $\mM_C$ consists of $C$ where $C=X,Z$.
\end{itemize}

\begin{prop}\label{pro:standard form}
Let $\sf{\overline{DPG}}_\mM$ denote the quotient double category obtained by imposing the rewrite rules above. Then, in this category any $\mM$-labeled double port graph $\Gamma$ can be represented by a unique standard form  $\Gamma_C{\circ}\Gamma_M{\circ}\Gamma_P$ where $\Gamma_A$ is a $\mM_A$-labeled double port graph with $A=P,M,C$.
\end{prop}
\begin{proof}
Proof is the same as \cite[Theorem 3]{danos2007measurement}.
\end{proof}
}

\begin{figure}[t]
\centering
\begin{tikzpicture}[x=0.5em,y=0.5em]
\def\xshift{-20}
\def\vshift{10}
\def\del{5}
\begin{scope}[scale=0.75,transform shape]

\Block{\xshift-\vshift}{0}{$J(\pi/4)$}{J}

\Block{\xshift+\vshift}{5}{$N_X$}{NU}
\Block{\xshift+\vshift}{-2}{$N_T$}{NM}
\Block{\xshift+\vshift}{-9}{$N_X$}{ND}

\Block{\xshift+2*\vshift}{8.5}{$E$}{E}
\Block{\xshift+3*\vshift}{1.5}{$E$}{EE}
\Block{\xshift+4*\vshift}{-5.5}{$E$}{EEE}

\Block{\xshift+5*\vshift+\del}{21+20}{$\mathbf{0}$}{OL}
\Block{\xshift+7.4*\vshift+\del}{21+20}{$\mathbf{0}$}{OR}
\Block{\xshift+5*\vshift+\del}{13+20}{$M_{X/Y}$}{ML}
\Block{\xshift+5*\vshift+\del}{5+20}{$c$}{C}

\Block{\xshift+6*\vshift+\del}{21+12}{$\mathbf{0}$}{OLL}
\Block{\xshift+6*\vshift+\del}{13+12}{$M_{X/Y}$}{MLL}
\Block{\xshift+6*\vshift+\del}{5+12}{$c$}{CC}

\Block{\xshift+5*\vshift+\del}{-2}{$\on{AND}$}{A}
\Block{\xshift+6*\vshift+\del}{7}{$M_{X/Y}$}{MR}

\Block{\xshift+6.8*\vshift+\del}{7}{$c$}{CCC}

\Block{\xshift+7*\vshift+\del}{-11}{$\on{XOR}$}{XORMT}
\Block{\xshift+7.4*\vshift+\del}{-19}{$\on{XOR}$}{XORMM}
\Block{\xshift+8*\vshift+\del}{-27}{$\on{XOR}$}{XORMD}

\Block{\xshift+7.4*\vshift+\del}{-35}{$X$}{X}
\Block{\xshift+6.6*\vshift+\del}{-35}{$Z$}{Z}

\HDraw{Z}{1}{1}{X}{1}{1}{0}{180}

\VDraw{CCC}{2}{2}{XORMD}{2}{2}{-90}{90}
\VDraw{XORMM}{1}{1}{XORMD}{2}{1}{-90}{90}
\VDraw{XORMT}{1}{1}{XORMM}{2}{1}{-90}{90}
\VDraw{MR}{1}{1}{XORMT}{2}{2}{-90}{90}
\VDraw{A}{1}{1}{XORMT}{2}{1}{-90}{90}
\VDraw{C}{2}{2}{Z}{1}{1}{-90}{90}
\VDraw{XORMD}{1}{1}{X}{1}{1}{-90}{90}
\VDraw{OR}{1}{1}{XORMM}{2}{2}{-90}{90}

\VDraw{CC}{2}{2}{CCC}{1}{1}{-90}{90}
\VDraw{CC}{2}{1}{MR}{1}{1}{-90}{90}
\VDraw{CCC}{2}{1}{A}{2}{2}{-90}{90}
\VDraw{C}{2}{1}{A}{2}{1}{-90}{90}

\VDraw{OLL}{1}{1}{MLL}{1}{1}{-90}{90}
\VDraw{MLL}{1}{1}{CC}{1}{1}{-90}{90}
\VDraw{OL}{1}{1}{ML}{1}{1}{-90}{90}
\VDraw{ML}{1}{1}{C}{1}{1}{-90}{90}

\HIndraw{E}{2}{1}{12}
\HDraw{E}{2}{1}{ML}{1}{1}{0}{180}
\HDraw{E}{2}{2}{EE}{2}{1}{0}{180}
\HDraw{EE}{2}{1}{MLL}{2}{1}{0}{180}
\HDraw{EE}{2}{2}{EEE}{2}{1}{0}{180}
\HDraw{EEE}{2}{1}{MR}{1}{1}{0}{180}
\HDraw{EEE}{2}{2}{Z}{1}{1}{0}{180}

\HDraw{NU}{1}{1}{E}{2}{2}{0}{180}
\HDraw{NM}{1}{1}{EE}{2}{2}{0}{180}
\HDraw{ND}{1}{1}{EEE}{2}{2}{0}{180}

\HOutdraw{X}{1}{1}{10}

\HIndraw{J}{1}{1}{3}
\HOutdraw{J}{1}{1}{3}

\path (-20,0) node {$=$};
\end{scope}
\end{tikzpicture}
\caption{{The teleported $J(\pi/4)$ gadget written as an $\mM[\pi/4]$-labeled double port graph. The boxes $N_X$ and $N_T$ denote state preparations, the boxes $E$ denote entangling operations, the boxes $M_{X/Y}$ denote destructive $X/Y$ measurements, the boxes $c$, $\on{AND}$, and $\on{XOR}$ denote Boolean operations in the classical control part, and the boxes $X$ and $Z$ denote the corresponding correction operations. Solid horizontal wires represent qubits, whereas dashed vertical wires represent classical bits.}}
\label{fig:teleported-j-pi-over-4}
\end{figure}

\end{document}